\documentclass[a4paper,11pt]{article}

\usepackage{fullpage}
\usepackage{times}
\usepackage{soul}
\usepackage{url}
\usepackage[utf8]{inputenc}
\usepackage[small]{caption}
\usepackage{graphicx}
\usepackage{xcolor}
\usepackage{amsmath}
\usepackage{booktabs}
\usepackage{algorithm}
\usepackage{algorithmic}
\usepackage{enumitem}
\usepackage{dsfont}
\usepackage{amsthm}
\usepackage{amssymb}
\usepackage{tikz}
\usepackage{bbm}

\newtheorem{theorem}{Theorem}[section]
\newtheorem{corollary}[theorem]{Corollary}
\newtheorem{proposition}[theorem]{Proposition}
\newtheorem{lemma}[theorem]{Lemma}

\theoremstyle{definition}
\newtheorem{definition}[theorem]{Definition}

\newcommand*\circled[1]{\tikz[baseline=(char.base)]{
            \node[shape=circle,draw,inner sep=2pt] (char) {#1};}}
            
\usepackage{natbib}

\allowdisplaybreaks

\usepackage{authblk}

\title{\bf Weighted Fair Division with Matroid-Rank Valuations: Monotonicity and Strategyproofness}

\author[1]{Warut Suksompong}
\author[2]{Nicholas Teh}

\affil[1]{National University of Singapore, Singapore}
\affil[2]{University of Oxford, UK}

\newcommand{\mwnwtie}{MWNW$^\text{tie}$}
\newcommand{\therule}{$\mathcal{R}_f$}
\newcommand{\ruletie}{$\mathcal{R}^\text{tie}_f$}
\newcommand{\alow}{\mathcal{L}}
\newcommand{\ahigh}{\mathcal{H}}
\newcommand{\truth}{\text{truth}}
\newcommand{\lie}{\text{lie}}

\date{\vspace{-10mm}}

\begin{document}

\maketitle

\begin{abstract}
    We study the problem of fairly allocating indivisible goods to agents with weights corresponding to their entitlements.
    Previous work has shown that, when agents have binary additive valuations, the maximum weighted Nash welfare rule is resource-, population-, and weight-monotone, satisfies group-strategyproofness, and can be implemented in polynomial time.
    We generalize these results to the class of weighted additive welfarist rules with concave functions and agents with matroid-rank (also known as binary submodular) valuations.    
\end{abstract}

\section{Introduction}
\label{sec:intro}

The problem of fairly allocating a set of scarce resources among interested parties is longstanding and fundamental in economics \citep{BramsTa96,Moulin03}.
Applications of fair division are wide-ranging, from allocating medical supplies to communities or schoolteachers to primary schools, to dividing assets in a divorce settlement and usage rights of a jointly invested facility.

While numerous approaches have been proposed and investigated in the fair division literature, an approach that has enjoyed significant attention is the \emph{maximum Nash welfare (MNW)} rule, which chooses an allocation that maximizes the product of the agents' utilities.
When the resource consists of indivisible goods and agents have additive valuations, \citet{CaragiannisKuMo19} showed that an MNW allocation satisfies the fairness notion of \emph{envy-freeness up to one good (EF1)} as well as the efficiency notion of \emph{Pareto-optimality}.
The remarkable fairness of MNW was further cemented when \citet{Suksompong23} characterized it as the unique rule within the class of \emph{additive welfarist rules} that always produces an EF1 allocation.
Additive welfarist rules select an allocation that maximizes a welfare notion given by the sum of some increasing function of the agents' utilities \citep[p.~67]{Moulin03}---in particular, MNW corresponds to taking the logarithm function.\footnote{\citet{YuenSu23} extended the characterization to all (not necessarily additive) welfarist rules.}

Although the majority of the fair division literature---including the aforementioned works of \citet{CaragiannisKuMo19} and \citet{Suksompong23}---assumes that all agents have the same entitlement to the resource, this assumption fails to hold in several practical scenarios.
For instance, when allocating resources among schools or communities of different sizes, larger schools and communities naturally deserve a larger proportion of the shared resource pool.
Similarly, in inheritance division, closer relatives are typically entitled to a larger share than distant ones.
As a consequence, several researchers have recently explored the fair allocation of indivisible goods when each agent has a \emph{weight} representing her entitlement \citep{FarhadiGhHa19,AzizMoSa20,AzizGaMi23,BabaioffEzFe21,BabaioffNiTa21,ChakrabortyIgSu21,ChakrabortyScSu21,ChakrabortySeSu22,SuksompongTe22,HoeferScVa23,ScarlettTeZi21}.
In the weighted setting, MNW can be generalized to the \emph{maximum weighted Nash welfare (MWNW)} rule, which selects an allocation maximizing the weighted product of the agents' utilities, where the weights appear in the exponents.
\citet{ChakrabortyIgSu21} showed that under additive valuations, MWNW satisfies Pareto-optimality as well as a weighted extension of EF1 called \emph{weak weighted envy-freeness up to one good}.
\citet{SuksompongTe22} proved that under \emph{binary} additive valuations---where each agent's utility for each good is either $0$ or $1$---MWNW is resource- and population-monotone, group-strategyproof, and can be computed in polynomial time.\footnote{To be more precise, their results hold for a version of MWNW with specific tie-breaking, which they called \mwnwtie.}
While these results established MWNW as a strong candidate rule for binary additive valuations, they still left open the possibility that more attractive rules exist in this domain.\footnote{Note that weighted fair division under binary additive valuations generalizes the well-studied setting of \emph{apportionment} \citep{BalinskiYo01,Pukelsheim14}, which corresponds to the case of identical goods (i.e., each agent's utility for each good is $1$).} 

Extending the definition of additive welfarist rules to the weighted setting, one can consider the class of \emph{weighted additive welfarist rules}, which select an allocation that maximizes a welfare notion given by the \emph{weighted} sum of some increasing function of the agents' utilities.
As with MNW in the unweighted setting, MWNW corresponds to taking the logarithm function in the weighted setting.
Recently, \citet[Sec.~6]{MontanariScSu22} proved that the \emph{maximum weighted harmonic welfare (MWHW)} rule, which takes a function based on the harmonic numbers, offers stronger guarantees in terms of weighted envy-freeness than MWNW.
More generally, they showed the same for a family of rules based on ``modified harmonic numbers''.
In fact, their findings continue to hold even for \emph{matroid-rank} valuations---also known as \emph{binary submodular} valuations---where the valuation function of each agent is submodular and the marginal utility for each good is always either $0$ or $1$.
Matroid-rank valuations generalize binary additive valuations and capture settings such as course allocation in universities or public housing allocation among ethnic groups \citep{BenabbouChIg21}; as such, they have received attention from a number of researchers in fair division \citep{BabaioffEzFe21-dichotomous,BarmanVe21,BarmanVe22,BenabbouChIg21,GokoIgKa22,MontanariScSu22,ViswanathanZi22,ViswanathanZi23}.

Our discussion thus far raises two natural questions:
do the results of \citet{SuksompongTe22} hold for other weighted additive welfarist rules besides MWNW, and do they continue to hold even for matroid-rank valuations?

\subsection{Our Contributions}
\label{sec:contribution}

In this paper, we provide positive answers to both of the above questions.
Specifically, we focus on weighted additive welfarist rules with a \emph{concave} function---this class of rules encompasses MWNW, MWHW, as well as all variations of MWHW based on modified harmonic numbers considered by \citet{MontanariScSu22}.
Such rules intuitively aim to distribute utilities fairly among agents: for instance, if all agents have binary additive valuations and equal weights and each agent has a value of~$1$ for each good, then a weighted additive welfarist rule with any concave function spreads the goods as equally as possible among the agents (whereas one with a convex function allocates all goods to the same agent, a patently unfair outcome in this scenario).\footnote{Interestingly, Theorem~3.14 of \citet{BenabbouChIg21} implies that in the \emph{unweighted} setting, all additive welfarist rules with a strictly concave function are equivalent to MNW.
However, as \citet{MontanariScSu22} showed, in the \emph{weighted} setting, different concave functions yield different weighted additive welfarist rules with varying fairness guarantees even for binary additive valuations.}
For a concave (and strictly increasing) function~$f$, we consider the weighted additive welfarist rule with function~$f$ with specific tie-breaking; we denote this rule by \ruletie{}.
We generalize the results of \citet{SuksompongTe22}, both in terms of the rules (from the logarithm function to an arbitrary concave function) and in terms of the valuation functions (from binary additive to matroid-rank valuations).
More specifically, our contributions are as follows.
\begin{itemize}
\item First, we show that \ruletie{} is resource-, population-, and weight-monotone under matroid-rank valuations. 
To the best of our knowledge, all three properties have not been studied in the matroid-rank setting.
In fact, our result on weight-monotonicity holds more generally and only requires the function~$f$ to be strictly increasing.\footnote{\citet{SuksompongTe22} did not consider weight-monotonicity. \citet{ChakrabortyScSu21} showed that MWNW satisfies weight-monotonicity under general additive valuations.}
\item Second, we show that \ruletie{} satisfies group-strategyproofness.
In addition to generalizing the corresponding result of \citet{SuksompongTe22} from the weighted binary additive setting, this also extends the work of \citet{BarmanVe22} from the unweighted matroid-rank setting.
\item Third, we exhibit an algorithm that computes an \ruletie{} allocation in polynomial time.
\end{itemize}
Taken together, these results imply that the desirable properties of MWNW established in prior work are in fact not unique to MWNW, but instead hold for other weighted additive welfarist rules as well.
In particular, combined with the aforementioned results of \citet{MontanariScSu22}, our findings demonstrate that MWHW and its variations are arguably stronger candidate rules for agents with arbitrary entitlements and matroid-rank (or binary additive) valuations.

\subsection{Further Related Work}

While several prior works have focused on the unweighted matroid-rank or the weighted (binary) additive setting as we mentioned earlier, the only two previous papers that considered the weighted matroid-rank setting (to our knowledge) are those by \citet{MontanariScSu22} and \citet{ViswanathanZi22}.
\citet{MontanariScSu22} investigated weighted envy-freeness notions---they proved that MWHW and its variants satisfy a notion called \emph{transferable WEF$(x,1-x)$}, which MWNW fails.
\citet{ViswanathanZi22} proposed a framework that can be used to compute an allocation satisfying a range of fairness objectives, including MWNW, in polynomial time.\footnote{\citet[App.~B]{MontanariScSu22} noted that the framework of \citet{ViswanathanZi22} can also handle MWHW and its variations.}
They also showed that the algorithms resulting from their framework satisfy (individual) strategyproofness.
We note that their techniques are based on those of \citet{BabaioffEzFe21-dichotomous}, which were used to show (individual) strategyproofness in the unweighted matroid-rank setting, and are therefore unlikely to be directly useful for establishing group-strategyproofness.
By contrast, we follow the approach of \citet{BarmanVe22} for exhibiting group-strategyproofness in the unweighted matroid-rank setting.
We also remark that neither \citet{MontanariScSu22} nor \citet{ViswanathanZi22} considered any of the monotonicity properties.

In the unweighted matroid-rank setting, an interesting rule is the \emph{prioritized egalitarian} mechanism put forward by \citet{BabaioffEzFe21-dichotomous}.
This rule is (individually) strategyproof and returns an allocation which is \emph{Lorenz-dominating}, meaning that it Lorenz-dominates every other allocation.\footnote{An allocation $\mathcal{A}$ is said to \emph{Lorenz-dominate} another allocation $\mathcal{B}$ if the smallest utility induced by $\mathcal{A}$ is at least as large as that induced by $\mathcal{B}$, the sum of the two smallest utilities induced by $\mathcal{A}$ is at least as large as that induced by $\mathcal{B}$, and so on.}
A Lorenz-dominating allocation maximizes the Nash welfare and satisfies \emph{envy-freeness up to any good (EFX)}; a key contribution of \citet{BabaioffEzFe21-dichotomous} was to establish that such an allocation always exists in their setting.
However, it is unclear how to extend the concept of Lorenz-domination to the weighted setting in such a way that existence is preserved.
For example, if we consider weighted utilities where each agent's utility is multiplied by her weight, then in an instance consisting of two agents with weight $1$ and~$2$, respectively, and two goods each yielding value~$1$ to each agent, no allocation is Lorenz-dominating.
Indeed, allocating both goods to the first agent yields weighted utilities $(2,0)$, allocating both goods to the second agent yields weighted utilities $(0,4)$, and allocating one good to each agent yields weighted utilities $(1,2)$.
A similar issue arises if we consider weighted utilities where each agent's utility is raised to an exponent corresponding to her weight.

\section{Preliminaries}
\label{sec:prelim}

In the allocation of indivisible goods, we are given a set $N = \{1,\dots,n\}$ of $n$ \textit{agents} and a set $G = \{g_1, \dots, g_m \}$ of $m$ \textit{goods}. 
Both $N$ and $G$ are not necessarily fixed, as extra agents or goods may be added.
Subsets of goods in $G$ are referred to as \textit{bundles}. Each agent $i \in N$ has a \textit{weight} $w_i > 0$ representing her entitlement, and a nonnegative \textit{valuation function} (or \emph{utility function}) $v_i$ over bundles of goods. 
The list $\mathbf{v} = (v_1,\dots,v_n)$ is called a \emph{valuation profile}.
The setting where all agents have the same weight is sometimes referred to as the \emph{unweighted setting}.
For convenience, we write $v_i(g)$ instead of $v_i(\{g\})$ for a single good~$g$.
We assume throughout the paper that $v_i$ is \emph{matroid-rank} (also known as \emph{binary submodular}); more precisely, $v_i$ satisfies the following properties:
\begin{itemize}
\item \emph{binary}: $v_i(G'\cup\{g\}) - v_i(G')\in \{0,1\}$ for all $G'\subseteq G$ and $g\in G\setminus G'$;
\item \emph{monotone}: $v_i(G') \le v_i(G'')$ for all $G'\subseteq G''\subseteq G$;
\item \emph{submodular}: $v_i(G'\cup\{g\}) - v_i(G') \ge v_i(G''\cup \{g\}) - v_i(G'')$ for all $G'\subseteq G''\subseteq G$ and $g\in G\setminus G''$;
\item \emph{normalized}: $v_i(\emptyset) = 0$.
\end{itemize}
An \emph{instance} consists of the set of agents~$N$, the set of goods~$G$, the agents' weights $(w_1,\dots,w_n)$, and the valuation profile $(v_1,\dots,v_n)$.
 
An \emph{allocation} $\mathcal{A} = (A_1, \dots, A_n)$ is a list of $n$ bundles such that no two bundles overlap, where agent~$i$ receives bundle $A_i$.
Note that it is not necessary that $\bigcup_{i\in N}A_i = G$.
For an allocation $\mathcal{A}$, let its \emph{utility vector} be $(v_1(A_1),\dots,v_n(A_n))$.
An \textit{allocation rule}, or simply a \emph{rule}, is a function that maps each instance to an allocation. 
Under matroid-rank valuations, an allocation $\mathcal{A}$ is said to be \emph{non-redundant}\footnote{The term \emph{non-redundant} was used by \citet{BabaioffEzFe21-dichotomous}.
The same concept has also been called \emph{clean} \citep{BenabbouChIg21} and \emph{non-wasteful} \citep{BarmanVe22}.} if $v_i(A_i) - v_i(A_i\setminus\{g\}) > 0$ (equivalently, $v_i(A_i) - v_i(A_i\setminus\{g\}) = 1$) for all $i\in N$ and $g\in A_i$.
It is known that $\mathcal{A}$ is non-redundant if and only if $v_i(A_i) = |A_i|$ for all $i\in N$ \citep[Prop.~3.3]{BenabbouChIg21}.
An allocation~$\mathcal{A}$ is \emph{Pareto-optimal} if there does not exist another allocation $\mathcal{A}'$ such that $v_i(A_i') \ge v_i(A_i)$ for all $i\in N$ and the inequality is strict for at least one $i\in N$; such an allocation~$\mathcal{A}'$ is said to \emph{Pareto-dominate} $\mathcal{A}$.
We denote by $N^+_\mathcal{A}\subseteq N$ the subset of agents receiving positive utility from~$\mathcal{A}$.
For a set of agents $S\subseteq N$, let  $\mathcal{A}_S$ denote the allocation derived from restricting $\mathcal{A}$ to the bundles of the agents in $S$ (in the same order as in $\mathcal{A})$. 

We can now state the definition of weighted additive welfarist rules.

\begin{definition}[\therule{}]
\label{def:Rf}
	Let $f : \mathbb{Z}_{\ge 0} \rightarrow [-\infty,\infty)$ be a strictly increasing function.
The \emph{weighted additive welfarist rule with function $f$}, denoted by \therule{}, chooses an allocation~$\mathcal{A}$ that maximizes the weighted welfare $\sum_{i\in N} w_i \cdot f(v_i(A_i))$.
If there are multiple such allocations, \therule{} may choose an arbitrary one.
The only exception is when the maximum possible weighted welfare is $-\infty$, in which case \therule{} first maximizes the number of agents $i\in N$ such that $f(v_i(A_i)) \neq -\infty$ (i.e., the number of agents receiving positive utility, $|N^+_\mathcal{A}|$),\footnote{Note that by definition of $f$, the only possible $k\in\mathbb{Z}_{\ge 0}$ such that $f(k) = -\infty$ is $k = 0$.} prioritizing agents with lower indices lexicographically in case of ties, then subject to that, maximizes the weighted welfare $\sum_{i \in N^+_\mathcal{A}} w_i \cdot f(v_i(A_i))$ among such agents.

An allocation is said to be an \emph{\therule{} allocation} if it can be chosen by the rule \therule{}.
\end{definition}

It follows from Definition~\ref{def:Rf} that an \therule{} allocation is always Pareto-optimal.
 
 Similarly to the additional tie-breaking specifications of the MNW and MWNW rules introduced by \cite{HalpernPrPs20} and \cite{SuksompongTe22}, respectively, we consider a specific version of the rule \therule{} called \ruletie{}.
 An allocation~$\mathcal{A}$ is said to be \emph{lexicographically dominating} within a set of allocations $\Pi$ if it maximizes the utility vector in a lexicographical order.
 More precisely, among all allocations in $\Pi$, the allocation $\mathcal{A}$ maximizes $v_1(A_1)$, then subject to that, maximizes $v_2(A_2)$, and so on.
 Then, we define \ruletie{} as follows.
 
 \begin{definition}[\ruletie{}] \label{def:ruletie}
 Given a function~$f$ as in Definition~\ref{def:Rf},
 	the rule \ruletie{} returns an allocation $\mathcal{A}$ such that
 	\begin{enumerate}
 		\item $\mathcal{A}$ is an \therule{} allocation that is also lexicographically dominating within the set of all \therule{} allocations; and
 		\item $\mathcal{A}$ is non-redundant.
 	\end{enumerate}
  An allocation is said to be an \emph{\ruletie{} allocation} if it can be chosen by the rule \ruletie{}.

 \end{definition}
 Clearly, there exists an allocation satisfying Condition~1 of Definition~\ref{def:ruletie}.
 Given such an allocation, one can find an allocation that additionally satisfies Condition~2 by iteratively removing a ``redundant'' good (i.e., a good whose removal does not decrease its owner's utility) until no such good exists.
 Hence, \ruletie{} is well-defined.
 If there are multiple allocations satisfying both conditions, \ruletie{} arbitrarily picks one of them.
 Note that even though there can be more than one \ruletie{} allocation, all such allocations have the same utility vector.
 As discussed in Section~\ref{sec:contribution}, we will focus on \ruletie{} rules where $f$ is concave, which means that
\begin{equation*}
		f(k+1) - f(k) \ge f(k+2) - f(k+1) \text { for all } k \geq 0.
\end{equation*}

Since the valuations that we consider in this paper may be non-additive, in order to reason about the running time of algorithms, we assume that an algorithm can query the value of any agent~$i$ for any bundle~$G'\subseteq G$ in constant time.
This \emph{value oracle} assumption is standard when dealing with non-additive valuations, including matroid-rank valuations \citep{BarmanVe21,BarmanVe21-XOS,BenabbouChIg21,GokoIgKa22,ViswanathanZi23}.

\subsection{Exchange Graphs and Path Augmentation}
\label{sec:exchange}

Next, following \citet{BarmanVe22}, we introduce several definitions and tools that will be useful for working with matroid-rank valuations.

Given two allocations $\mathcal{A} = (A_1,\dots,A_n)$ and $\mathcal{A}' = (A'_1,\dots,A'_n)$, let $\alow(\mathcal{A}',\mathcal{A}) :=\{i \in N : |A'_i| < |A_i|\}$ and $\ahigh(\mathcal{A}',\mathcal{A}) :=\{i \in N : |A'_i| > |A_i|\}$.
An \emph{exchange graph} of an allocation $\mathcal{A}$ is defined as the directed graph $\mathcal{G}(\mathcal{A}) = (G,E)$, with the set of vertices $G$ (i.e., each vertex corresponds to a good), and a directed edge $(g,g') \in E$ if and only if for some $i \in N$ it holds that $g \in A_i$, $g' \not\in A_i$, and $v_i(A_i) = v_i(A_i \cup \{g'\} \setminus \{g\})$.
For any directed path $P = (g_1,\dots,g_t)$ in the exchange graph $\mathcal{G}(\mathcal{A})$ and any $i\in N$,  define the bundle $A_i \triangle P$, which is the result of ``augmenting'' the bundle~$A_i$ along the path~$P$, as follows: starting from $A_i$, for each $j\in \{1,2,\dots,t-1\}$, if $g_j\in A_i$, replace $g_j$ with $g_{j+1}$.
Moreover, for any bundle $B$ and any $i\in N$, let $F_i(B)$ be the set of goods $g\not\in B$ such that $v_i(B\cup\{g\}) - v_i(B) = 1$.

We now state three lemmas on exchange graphs and path augmentation from prior work.
The first lemma was due to \citet{Schrijver03} and stated as Lemma~2.4 in the work of \citet{BarmanVe21}; it establishes the changes in bundle size of agents involved in a path augmentation.\footnote{The lemma is also stated as Lemma~1 and proved in the extended version of the work of \citet{BarmanVe21}.}
\begin{lemma} [\citep{Schrijver03,BarmanVe21}] \label{lemma:gsp_schrijver}
    Let $\mathcal{X}$  be any non-redundant allocation. 
    For any two agents $i,j \in N$, let $Q = (g_1,\dots,g_t)$ be a shortest path in the exchange graph $\mathcal{G}(\mathcal{X})$ between the vertex sets $F_i(X_i)$ and $X_j$ (in particular, $g_1 \in F_i(X_i)$ and $g_t \in X_j$). 
    Then, we have
    \begin{itemize}
    \item $v_i((X_i \triangle Q) \cup \{g_1\}) = |X_i| + 1$;
    \item $v_j(X_j \setminus \{g_t\}) = |X_j| - 1$; and
    \item $v_k(X_k \triangle Q) = |X_k|$ for all $k \in N \setminus \{i,j\}$.
    \end{itemize}      
\end{lemma}
Intuitively, Lemma~\ref{lemma:gsp_schrijver} indicates that if we perform augmentation along a shortest path between $F_i(X_i)$ and $X_j$, add $g_1$ to $i$'s bundle, and remove $g_t$ from $j$'s bundle, we obtain a new non-redundant allocation in which agent $i$'s utility increases by $1$, agent $j$'s utility decreases by $1$, and the utilities of the other agents remain unchanged.

The second lemma was stated as a consequence of Lemma 6 in the extended version of the work of \citet{BarmanVe22}. 
\begin{lemma}[\citep{BarmanVe22}] \label{lemma:gsp_BarmanVerma}
    Let $\mathcal{X} = (X_1,\dots,X_n)$ be a non-redundant allocation and $\mathcal{A} = (A_1,\dots,A_n)$ be a Pareto-optimal non-redundant allocation such that there exists an agent $h \in \ahigh(\mathcal{X},\mathcal{A})$. Then, there exists an agent $\ell \in \alow(\mathcal{X},\mathcal{A})$ such that
    \begin{enumerate}
    \item Starting from $\mathcal{X}$, we can obtain a non-redundant allocation in which agent~$\ell$'s utility increases by~$1$, agent~$h$'s utility decreases by $1$, and every other agent receives the same utility as before;
    \item Complementarily, starting from $\mathcal{A}$, we can obtain a non-redundant allocation in which agent~$\ell$'s utility decreases by~$1$, agent~$h$'s utility increases by~$1$, and every other agent receives the same utility as before.
    \end{enumerate}
\end{lemma}

The third lemma was stated as Lemma~2 in the work of \cite{BarmanVe22}.\footnote{This is also Lemma~6 in the extended version of their work.}
The second condition is the same as in Lemma~\ref{lemma:gsp_BarmanVerma} above, but the first condition guarantees the existence of a path in the exchange graph rather than a non-redundant allocation.
\begin{lemma}[\citep{BarmanVe22}] \label{lemma:gsp_BarmanVerma_path}
    Let $\mathcal{X} = (X_1,\dots,X_n)$ be a non-redundant allocation and $\mathcal{A} = (A_1,\dots,A_n)$ be a Pareto-optimal non-redundant allocation such that there exists an agent $h \in \ahigh(\mathcal{X},\mathcal{A})$.
    Then, there exists an agent $\ell \in \alow(\mathcal{X},\mathcal{A})$ along with a simple directed path $P = (g_k, g_{k-1}, \dots, g_2, g_1)$ in the exchange graph $\mathcal{G}(\mathcal{X})$ such that 
    \begin{enumerate}
    \item The source vertex is $g_k \in A_\ell \cap F_\ell(X_\ell)$ and the sink vertex is $g_1 \in X_h$;
    \item Starting from $\mathcal{A}$, we can obtain a non-redundant allocation in which agent~$\ell$'s utility decreases by~$1$, agent~$h$'s utility increases by~$1$, and every other agent receives the same utility as before.
    \end{enumerate}    
\end{lemma}

\section{Resource-Monotonicity}
\label{sec:res-mon}

In this section, we consider resource-monotonicity, an intuitive property which states that when an extra good is added, no agent should receive lower utility from the allocation output by the rule.
\begin{definition}[Resource-monotonicity]
    An allocation rule $\mathcal{F}$ is \emph{resource-monotone} if the following holds: For any two instances $\mathcal{I}$ and $\mathcal{I}'$ such that $\mathcal{I}'$ can be obtained from $\mathcal{I}$ by adding one extra good, if $\mathcal{F}(\mathcal{I}) = \mathcal{A}$ and $\mathcal{F}(\mathcal{I}') = \mathcal{A}'$, then for each $i \in N$, $v_i(A_i) \leq v_i(A'_i)$.
\end{definition}

\cite{ChakrabortyScSu21} showed that under general additive valuations, MNW fails to be resource-monotone in the unweighted setting regardless of tie-breaking. 
\cite{SuksompongTe22} subsequently showed that the MWNW rule with lexicographical tie-breaking, MWNW$^\text{tie}$, satisfies resource-monotonicity under binary additive valuations in the weighted setting.
We will generalize this positive result to \ruletie{} for any concave function $f$ and matroid-rank valuations.

We first make the following observation.

\begin{lemma} \label{lemma:resmon_agenthigher}
	Let $\mathcal{I}$ and $\mathcal{I}'$ be instances such that $\mathcal{I}'$ can be obtained from $\mathcal{I}$ by adding one extra good~$g$, and let $f$ be a strictly increasing function.
 Suppose that \ruletie{}$(\mathcal{I}) = \mathcal{A}$ and \ruletie{}$(\mathcal{I}') = \mathcal{A}'$.
	If $\alow(\mathcal{A}', \mathcal{A}) \neq \emptyset$, then $\ahigh(\mathcal{A}', \mathcal{A}) \neq \emptyset$ as well.
\end{lemma}
\begin{proof}
If $\alow(\mathcal{A}', \mathcal{A}) \neq \emptyset$ but $\ahigh(\mathcal{A}', \mathcal{A}) = \emptyset$, then $\mathcal{A}$ Pareto-dominates $\mathcal{A}'$.
On the other hand, since \ruletie{}$(\mathcal{I}') = \mathcal{A}'$, the allocation $\mathcal{A}'$ is Pareto-optimal in $\mathcal{I}'$, contradicting the fact that $\mathcal{A}$ is also a valid allocation in $\mathcal{I}'$ and Pareto-dominates $\mathcal{A}'$.
\end{proof}

We now establish the resource-monotonicity of the rule \ruletie{}.

\begin{theorem} \label{thm:resource_monotonicity}
    Under matroid-rank valuations, the rule \ruletie{} with any concave function $f$ is resource-monotone.
\end{theorem}

\begin{proof}
    Assume that $f$ is concave.
    Let $\mathcal{A}$ be an \ruletie{} allocation in an instance $\mathcal{I}$, and let $N$ and $G$ denote the set of agents and the set of goods in $\mathcal{I}$, respectively.
    Suppose that an instance $\mathcal{I}'$ can be obtained from $\mathcal{I}$ by adding an extra good $g^*$. 

    First, assume that there exists an \ruletie{} allocation $\mathcal{A}'$ in $\mathcal{I}'$ such that $g^* \notin A'_i$ for all $i \in N$.
    Since $\mathcal{A}$ and $\mathcal{A}'$ are both \ruletie{} allocations, we have $v_i(A_i) = v_i(A'_i)$  for all $i\in N$.
    Since each agent receives the same utility from all \ruletie{} allocations in a given instance, resource-monotonicity holds between $\mathcal{I}$ and $\mathcal{I}'$.

    Assume from now on that for every \ruletie{} allocation $\mathcal{A}'$ in $\mathcal{I}'$ we have $g^* \in A'_i$ for some $i \in N$, and fix an \ruletie{} allocation $\mathcal{A}'$ in $\mathcal{I}'$.   
    Suppose for a contradiction that $|A_i| > |A'_i|$ for some agent $i \in N$. 
    Then, by Lemma \ref{lemma:resmon_agenthigher}, there must exist some other agent $j \in N \setminus \{i\}$ such that $|A_j| < |A'_j|$.
    Let $S \subseteq N$ be the set of all agents that experienced an increase in bundle value (which is the same as bundle size due to non-redundancy), i.e., for every agent $j \in S$,     \begin{equation*}
        |A_j| < |A'_j|.
    \end{equation*}    
    
    Suppose that an instance $\mathcal{I}^*$ can be obtained from $\mathcal{I}'$ by adding\footnote{For notational convenience, we use $0$ as an agent number.} an extra agent $0$ with an arbitrary positive weight, and let $N^* = N\cup\{0\}$.
    In $\mathcal{I}^*$, let $\mathcal{A}^*$ be the allocation that is the same as $\mathcal{A}$ except that there is an additional bundle $A^*_0 = \{g^*\}$, and let $\mathcal{A}'^*$  be the allocation that is the same as $\mathcal{A}'$ except that there is an additional bundle $A'^*_0 = \emptyset$.
    Also, let the valuation $v_0$ of agent~$0$ be a binary additive function such that
    \begin{equation*}
        v_0(g^*)=1 \text{ and } v_0(g) = 0 \text{ for all } g \in G.
    \end{equation*}
    Since an allocation is non-redundant if and only if each agent has the same value for her bundle as the bundle size, and both $\mathcal{A}$ and $\mathcal{A}'$ are non-redundant, $\mathcal{A}^*$ and $\mathcal{A}'^*$ are also non-redundant.
    Moreover, since $\mathcal{A}'$ is an \ruletie{} allocation in $\mathcal{I}'$, and every \ruletie{} allocation in $\mathcal{I}'$ allocates the good~$g^*$ to one of the agents in~$N$, the allocation~$\mathcal{A}'^*$ is Pareto-optimal in $\mathcal{I}^*$.
    Observe also that for each $k \in N$,
    \begin{equation}  \label{lem:resmon_same}
        |A_k| = |A_k^*| \text{ and } |A'_k| = |A'^*_k|.
    \end{equation}
    Further, we have $0 \in \alow(\mathcal{A}'^*,\mathcal{A}^*)$ and  $0 \notin \ahigh(\mathcal{A}'^*,\mathcal{A}^*)$.   

    Let $\ell\in N$ be an agent such that $|A_\ell| > |A'_\ell|$.
    By (\ref{lem:resmon_same}), we have that $|A^*_\ell| > |A'^*_\ell|$.
    This means that $\ell \in \ahigh(\mathcal{A}^*,\mathcal{A}'^*)$, where $\ell \neq 0$.
    Since $\mathcal{A}^*$ is non-redundant and $\mathcal{A}'^*$ is Pareto-optimal and non-redundant in $\mathcal{I}^*$, by the second property of Lemma \ref{lemma:gsp_BarmanVerma},\footnote{Note here that the roles of agents~$\ell$ and $h$ are reversed from those in Lemma~\ref{lemma:gsp_BarmanVerma}.} there exists an agent $h \in \alow(\mathcal{A}^*,\mathcal{A}'^*)$ such that from $\mathcal{A}'^*$, we can obtain a non-redundant allocation $\widehat{\mathcal{A}}'^*$ such that
    \begin{equation*}
        |\widehat{A}'^*_\ell|= |A'^*_\ell|+1, |\widehat{A}'^*_h|= |A'^*_h| - 1,  \text{ and } |\widehat{A}'^*_k| = |A'^*_k| \text{ for all } k \in N^* \setminus \{\ell,h\}.
    \end{equation*}
    Note that $h \in \ahigh(\mathcal{A}'^*,\mathcal{A}^*)$ and $\ell \in \alow(\mathcal{A}'^*,\mathcal{A}^*)$.
    Since $0 \notin \ahigh(\mathcal{A}'^*,\mathcal{A}^*)$, we have $h\neq 0$.
    Since $\ell \neq 0$ and $h \neq 0$, we have that $|\widehat{A}'^*_0| = |A'^*_0|$.
    This means that from $\mathcal{A}'$, we can similarly obtain a non-redundant allocation $\widehat{\mathcal{A}}'$ in $\mathcal{I}'$
     such that 
    \begin{equation*}
        |\widehat{A}'_\ell| = |A'_\ell|+1, |\widehat{A}'_h| = |A'_h| - 1, \text{ and } |\widehat{A}'_k| = |A'_k| \text{ for all } k \in N \setminus \{\ell,h\}.
    \end{equation*}
    
    Complementarily, by the first property of Lemma \ref{lemma:gsp_BarmanVerma}, from $\mathcal{A}^*$ we can obtain an allocation $\widehat{\mathcal{A}}^*$ such that
    \begin{equation*}
        |\widehat{A}^*_\ell| = |A^*_\ell|-1, |\widehat{A}^*_h| = |A^*_h|+1, \text{ and } |\widehat{A}^*_k| = |A^*_k| \text{ for all } k \in N^* \setminus \{\ell,h\}.
    \end{equation*}
    Since $\ell \ne 0$ and $h \neq 0$, we have that $|\widehat{A}^*_0| = |A^*_0|$. 
    This means that from $\mathcal{A}$, we can similarly obtain a non-redundant allocation $\widehat{\mathcal{A}}$ in $\mathcal{I}$ such that 
        \begin{equation*}
        |\widehat{A}_\ell| = |A_\ell|-1, |\widehat{A}_h| = |A_h|+1, \text{ and } |\widehat{A}_k| = |A_k| \text{ for all } k \in N \setminus \{\ell,h\}.
    \end{equation*}
    Moreover, since $\ell \in \alow(\mathcal{A}'^*, \mathcal{A}^*)$ and $\ell \neq 0$, 
    \begin{equation} \label{eqn:resmon_tiebreak_L}
        |A'_\ell| < |A_\ell|.
    \end{equation}
    Similarly, since $h \in \ahigh(\mathcal{A}'^*,\mathcal{A}^*)$ and $h\ne 0$,
    \begin{equation} \label{eqn:resmon_tiebreak_H}
        |A'_h| > |A_h|.
    \end{equation}
    
    Assume first that the function $f$ satisfies $f(0) = -\infty$. 
    We prove four claims.
    \begin{description}
        \item[Claim 1: $|A_h| \ge 1$.]
        Assume for a contradiction that $|A_h| = 0$.
        We know from (\ref{eqn:resmon_tiebreak_L}) that $|A_\ell| \geq 1$. If $|A_\ell| > 1$, then the allocation $\widehat{\mathcal{A}}$ with  utility vector\footnote{Even though we write agent~$\ell$'s utility before agent~$h$'s in the utility vector, the actual order would be reversed if $h < \ell$.} $(|A_\ell| - 1, |A_h| + 1)$ for agents $\ell$ and $h$ has strictly more agents receiving positive utility than $\mathcal{A}$, thereby contradicting the fact that $\mathcal{A}$ is an \ruletie{} allocation in $\mathcal{I}$.
        Thus, $|A_\ell| = 1$, which by (\ref{eqn:resmon_tiebreak_L}) means $|A'_\ell| = 0$.
        Since the utility vector $(|A_\ell|, |A_h|)$ is preferred to $(|A_\ell|-1, |A_h|+1)$ because $\mathcal{A}$ is an \ruletie{} allocation in $\mathcal{I}$, it must be that $\ell < h$.
        
        On the other hand, we know from (\ref{eqn:resmon_tiebreak_H}) that $|A'_h| \geq 1$. 
        If $|A'_h| > 1$, then the allocation $\mathcal{\widehat{A}}'$ with  utility vector $(|A'_\ell| + 1, |A'_h| - 1)$ for agents $\ell$ and $h$ has strictly more agents receiving positive utility than $\mathcal{A}'$, thereby contradicting the fact that $\mathcal{A}'$ is an \ruletie{} allocation in $\mathcal{I}'$. 
        Thus, $|A'_h| = 1$.
        Since the utility vector $(|A'_\ell|, |A'_h|)$ is preferred to $(|A'_\ell| + 1, |A'_h| - 1)$ because $\mathcal{A}'$ is an \ruletie{} allocation in $\mathcal{I}'$, we must have that $h < \ell$, a contradiction with the previous paragraph.
        Hence, it must be that $|A_h| \geq 1$.

        \item[Claim 2: $|A'_h| \ge 2$.]
        This follows immediately from (\ref{eqn:resmon_tiebreak_H}) and Claim~1.

        \item[Claim 3: $|A'_\ell| \ge 1$.] 
        Assume for a contradiction that $|A'_\ell| = 0$.
        By Claim~2, the allocation $\mathcal{\widehat{A}}'$ with  utility vector $(|A'_\ell| + 1, |A'_h| - 1)$ for agents $\ell$ and $h$ has strictly more agents receiving positive utility than $\mathcal{A}'$, thereby contradicting the fact that $\mathcal{A}'$ is an \ruletie{} allocation in $\mathcal{I}'$. 
        Hence, it must be that $|A'_\ell| \ge 1$.

        \item[Claim 4: $|A_\ell| \ge 2$.]
        This follows immediately from (\ref{eqn:resmon_tiebreak_L}) and Claim~3.
    \end{description}
    Thus, we can assume henceforth that either $f(0) \neq -\infty$, or if $f(0) = -\infty$, then
    \begin{equation*}
        |A_\ell| \geq 2, \quad |A_h| \geq 1, \quad |A'_\ell| \geq 1, \quad |A'_h| \geq 2.
    \end{equation*}
    
    Since $\mathcal{A}'$ is an \ruletie{} allocation in $\mathcal{I}'$, and the rule \ruletie{} could have chosen the allocation $\widehat{\mathcal{A}}'$, it must be that
    \begin{equation} \label{resmon_newgoodaugment}
    \begin{split}
        w_\ell \cdot f(|A'_\ell|) + w_h \cdot f(|A'_h|) & \geq w_\ell \cdot f(|\widehat{A}'_\ell|) + w_h \cdot f(|\widehat{A}'_h|) \\
        & = w_\ell \cdot f(|A'_\ell|+1) + w_h \cdot f(|A'_h|-1),
    \end{split}
    \end{equation}
    where equality holds only if $h < \ell$.
    Note that by our claims, if $f(0) = -\infty$, then none of the terms $|A'_\ell|, |A'_h|, |A'_\ell| + 1, |A'_h| - 1$ can be $0$.
    
    Similarly, since $\mathcal{A}$ is an \ruletie{} allocation in $\mathcal{I}$, and the rule \ruletie{} could have chosen the allocation $\widehat{\mathcal{A}}$, 
    \begin{equation} \label{eqn:resmon_originalaugment}
        \begin{split}
            w_\ell \cdot f(|A_\ell|) + w_h \cdot f(|A_h|) & \geq w_\ell \cdot f(|\widehat{A}_\ell|) + w_h \cdot  f(|\widehat{A}_h|) \\
            & = w_\ell \cdot f(|A_\ell|-1) + w_h \cdot  f(|A_h|+1),
        \end{split}
    \end{equation}
    where equality holds only if $\ell < h$.
    Note that by our claims, if $f(0) = -\infty$, then none of the terms $|A_\ell|, |A_h|, |A_\ell| - 1, |A_h| + 1$ can be $0$.
    
    Now, by (\ref{eqn:resmon_tiebreak_L}) and the concavity of $f$, we have
    \begin{equation*} 
        f(|A'_\ell|+1) - f(|A'_\ell|) \ge f(|A_\ell|) - f(|A_\ell|-1),
    \end{equation*}
    which means that
    \begin{equation} \label{resmon_l_concavity}
        f(|A_\ell|-1) - f(|A_\ell|) \ge f(|A'_\ell|) - f(|A'_\ell|+1).
    \end{equation}
    Similarly, by (\ref{eqn:resmon_tiebreak_H}) and the concavity of $f$, we have
    \begin{equation} \label{resmon_h_concavity}
        f(|A_h|+1) - f(|A_h|) \ge f(|A'_h|) - f(|A'_h|-1).
    \end{equation}
    It follows that
    \begin{align*}
        0 & \geq w_\ell \cdot f(|A_\ell|-1) - w_\ell \cdot f(|A_\ell|) + w_h \cdot  f(|A_h|+1) - w_h \cdot f(|A_h|) \\
        & = w_\ell \cdot [f(|A_\ell|-1) -  f(|A_\ell|)] + w_h \cdot  [f(|A_h|+1) -  f(|A_h|)] \\
        & \ge w_\ell \cdot [f(|A'_\ell|) - f(|A'_\ell|+1)] + w_h\cdot [f(|A'_h|) - f(|A'_h| - 1)] \\
        & = w_\ell \cdot f(|A'_\ell|) - w_\ell \cdot f(|A'_\ell|+1) + w_h \cdot f(|A'_h|) - w_h \cdot f(|A'_h| - 1) \\
        &  \geq 0,
    \end{align*}
    where the first inequality holds by (\ref{eqn:resmon_originalaugment}), the second inequality by (\ref{resmon_l_concavity}) and (\ref{resmon_h_concavity}), and the last inequality by~(\ref{resmon_newgoodaugment}).
    Note also that the first inequality can be an equality only if $\ell < h$ while the last inequality can be an equality only if $h < \ell$.
    Hence, at least one of these two inequalities must be strict, which gives us the desired contradiction.    
\end{proof}

\section{Algorithm}
\label{sec:algo}

The resource-monotonicity of \ruletie{} from Section~\ref{sec:res-mon} provides us with an important tool for establishing the existence of a polynomial-time algorithm to compute \ruletie{} allocations.
Unlike in the binary additive case \citep{SuksompongTe22}, a more intricate argument is required when agents have matroid-rank valuations.
We begin with a few lemmas.
Note that we use the term ``nondegenerate path'' to mean a path containing at least one edge.
Recall also the definitions of $F_i$ and $\mathcal{G}$ from Section~\ref{sec:exchange}.

\begin{lemma} \label{lem:correctness_nopath}
    Let $\mathcal{I}$ and $\mathcal{I}'$ be instances with a set of agents~$N$ such that $\mathcal{I}'$ can be obtained from $\mathcal{I}$ by adding one extra good~$g$, and let \ruletie{}$(\mathcal{I}) = \mathcal{A}$.
    Suppose that an instance $\mathcal{I}^*$ can be obtained from $\mathcal{I}'$ by adding an extra agent $0$ with an arbitrary positive weight and a binary additive valuation function such that $v_0(g) = 1$ and $v_0(\overline{g}) = 0$ for all $\overline{g} \neq g$.
    Let the allocation $\widehat{\mathcal{A}}$ in $\mathcal{I}^*$ be the same as $\mathcal{A}$ except that there is an additional bundle $\widehat{A}_0 = \{g\}$.
    
    If $g \notin F_i(\widehat{A}_i)$ for all $i \in N$ and there does not exist a nondegenerate path in the exchange graph $\mathcal{G}(\widehat{\mathcal{A}})$ that starts at a good in $F_j(\widehat{A}_j)$ for some $j\in N$ and ends at~$g$, then $\mathcal{A}$ is an \ruletie{} allocation in $\mathcal{I}'$.
\end{lemma}

\begin{proof}
    Let \ruletie{}$(\mathcal{I}')= \mathcal{A}'$.
    If there exists an \ruletie{} allocation $\mathcal{A}''$ in $\mathcal{I}'$ such that $g \notin A''_i$ for all $i \in N$, then $\mathcal{A}$ is an \ruletie{} allocation in $\mathcal{I}'$ as well.
    Thus, assume that for every \ruletie{} allocation $\mathcal{A}''$ in $\mathcal{I}'$ we have $g \in A''_i$ for some $i \in N$.
    In particular, $g \in A'_i$ for some $i \in N$.
    Suppose that $g \notin F_i(\widehat{A}_i)$ for all $i \in N$ and there does not exist a nondegenerate path in the exchange graph $\mathcal{G}(\widehat{\mathcal{A}})$ that starts at a good in $F_j(\widehat{A}_j)$ for some $j\in N$ and ends at~$g$.
    We will show that this situation is, in fact, impossible.

    Let $\widehat{\mathcal{A}}'$  be the allocation that is the same as $\mathcal{A}'$ except that there is an additional bundle $\widehat{A}'_0 = \emptyset$.
    Since an allocation is non-redundant if and only if each agent has the same value for her bundle as the bundle size, and both $\mathcal{A}$ and $\mathcal{A}'$ are non-redundant, $\widehat{\mathcal{A}}$ and $\widehat{\mathcal{A}}'$ are also non-redundant.
    Moreover, since $\mathcal{A}'$ is an \ruletie{} allocation in $\mathcal{I}'$, and every \ruletie{} allocation in $\mathcal{I}'$ allocates the good~$g$ to one of the agents in~$N$, the allocation~$\widehat{\mathcal{A}}'$ is Pareto-optimal in $\mathcal{I}^*$.

    Now, since $0 \in \ahigh(\widehat{\mathcal{A}}, \widehat{\mathcal{A}}')$, by Lemma~\ref{lemma:gsp_BarmanVerma_path}, there exists an agent $j \in \alow(\widehat{\mathcal{A}},\widehat{\mathcal{A}}')$ along with a path in $\mathcal{G}(\widehat{\mathcal{A}})$ starting from a good in $F_j(\widehat{A}_j)$ and ending at $g$, the only good in $\widehat{A}_0$.
    Note that $j \ne 0$ because $j \in \alow(\widehat{\mathcal{A}},\widehat{\mathcal{A}}')$ and $0 \in \ahigh(\widehat{\mathcal{A}}, \widehat{\mathcal{A}}')$.
    Since $g \notin F_k(\widehat{A}_k)$ for all $k\in N$, this path is nondegenerate.
    This contradicts the assumption that there does not exist a nondegenerate path in $\mathcal{G}(\widehat{\mathcal{A}})$ that starts at a good in $F_j(\widehat{A}_j)$ and ends at~$g$.
\end{proof}

\begin{lemma} \label{lemma:resmon_bound_oneagent}
Consider an instance $\mathcal{I}$ with a set of agents $N\cup\{0\}$, where agent~$0$ has an arbitrary positive weight and a binary additive valuation function such that $v_0(g) = 1$ for some good $g$ and $v_0(\overline{g}) = 0$ for all $\overline{g} \neq g$.
	Let $\mathcal{X} = (X_0,X_1,\dots,X_n)$ be a non-redundant allocation and $\mathcal{A} = (A_0,A_1,\dots,A_n)$ be a Pareto-optimal non-redundant allocation, and suppose that $|X_0|=0$.
 If there exists an agent $h \in N$ with $|X_h| > |A_h| + 1$, or if there exist two distinct agents $h, h' \in N$ belonging to $\ahigh(\mathcal{X},\mathcal{A})$, then there exists an agent in $\alow(\mathcal{X},\mathcal{A})$ different from agent~$0$.
\end{lemma}
\begin{proof}
	Since $h \in \ahigh(\mathcal{X},\mathcal{A})$, by the first property of Lemma \ref{lemma:gsp_BarmanVerma}, there exists an agent $\ell \in \alow(\mathcal{X},\mathcal{A})$ and a non-redundant allocation $\mathcal{A}^*$ such that
    \begin{equation*}
        |A^*_\ell| = |X_\ell| + 1, |A^*_h| = |X_h|-1, \text{ and } |A^*_k| = |X_k| \text{ for all } k \in (N \cup \{0\}) \setminus \{\ell,h\}.
    \end{equation*}
    If $\ell \neq 0$, we are done. 
    Assume therefore that $\ell = 0$.

    First, consider the case where $|X_h| > |A_h| + 1$.
    We have that $|A^*_h| = |X_h| - 1 > |A_h|$. 
    Since $h \in \ahigh(\mathcal{A}^*,\mathcal{A})$, by the first property of Lemma \ref{lemma:gsp_BarmanVerma}, there exists an agent $\ell' \in \alow(\mathcal{A}^*,\mathcal{A})$ and a non-redundant allocation $\widehat{\mathcal{A}}^*$ such that
    \begin{equation*}
        |\widehat{A}^*_{\ell'}| = |A^*_{\ell'}| + 1, |\widehat{A}^*_h| = |A^*_h|-1, \text{ and } |\widehat{A}^*_k| = |A^*_k| \text{ for all } k \in (N \cup \{0\}) \setminus \{\ell',h\}.
    \end{equation*}
    If $\ell' = 0$, then
    $|\widehat{A}^*_0| = |A^*_0| + 1 = |X_0| + 2 = 2$, which is a contradiction since $\widehat{\mathcal{A}}^*$ is non-redundant and agent $0$ only values good $g$ positively.
    Thus, $\ell' \neq 0$, and so $\ell' \neq \ell$.
    Now, $\ell' \in \alow(\mathcal{A}^*,\mathcal{A})$ implies that $|X_{\ell'}| = |A^*_{\ell'}| < |A_{\ell'}|$, where the equality follows from the fact that $h \in \ahigh(\mathcal{A}^*,\mathcal{A})$ and $\ell' \notin\{\ell, h\}$.
    Hence, $\ell' \in \alow(\mathcal{X},\mathcal{A})$, which means that there exists an agent in $\alow(\mathcal{X},\mathcal{A})$ different from agent~$0$, as desired.

    Next, consider the case where $h, h' \in \ahigh(\mathcal{X},\mathcal{A})$ and $h\ne h'$.
    Since $h' \in \ahigh(\mathcal{X},\mathcal{A})$ and $\ell \in \alow(\mathcal{X},\mathcal{A})$, we have that $h' \neq \ell$.
    Hence, $|A^*_{h'}| = |X_{h'}| > |A_{h'}|$, which means that $h' \in \ahigh(\mathcal{A}^*,\mathcal{A})$.
    By the first property of Lemma \ref{lemma:gsp_BarmanVerma}, there exists an agent $\ell'' \in \alow(\mathcal{A}^*,\mathcal{A})$ and a non-redundant allocation $\overline{\mathcal{A}}^*$ where
    \begin{equation*}
        |\overline{A}^*_{\ell''}| = |A^*_{\ell''}| + 1, |\overline{A}^*_{h'}| = |A^*_{h'}|-1, \text{ and } |\overline{A}^*_k| = |A^*_k| \text{ for all } k \in (N \cup \{0\}) \setminus \{\ell'',h'\}.
    \end{equation*}
    If $\ell'' = 0$, then
    $|\overline{A}^*_0| = |A^*_0| + 1 = |X_0| + 2 = 2$, which is a contradiction since $\overline{\mathcal{A}}^*$ is non-redundant and agent~$0$ only values good $g$ positively.
    Thus, $\ell'' \neq 0$, and so $\ell'' \neq \ell$.
    Now, $\ell'' \in \alow(\mathcal{A}^*,\mathcal{A})$ implies that $|X_{\ell''}| = |A^*_{\ell''}| < |A_{\ell''}|$, where the equality follows from the fact that $h \in \ahigh(\mathcal{A}^*,\mathcal{A})$ and $\ell'' \notin \{\ell,h\}$.
    Hence, $\ell'' \in \alow(\mathcal{X},\mathcal{A})$, which again means that there exists an agent in $\alow(\mathcal{X},\mathcal{A})$ different from agent~$0$.
\end{proof}
 
\begin{lemma} \label{lem:correctness_gotpath}
    Let $\mathcal{I}$ and $\mathcal{I}'$ be instances with a set of agents~$N$ such that $\mathcal{I}'$ can be obtained from $\mathcal{I}$ by adding one extra good~$g$, and let \ruletie{}$(\mathcal{I}) = \mathcal{A}$ and \ruletie{}$(\mathcal{I}')= \mathcal{A}'$, where $f$ is a concave function.
    Suppose that an instance $\mathcal{I}^*$ can be obtained from $\mathcal{I}'$ by adding an extra agent $0$ with an arbitrary positive weight and a binary additive valuation function such that $v_0(g) = 1$ and $v_0(\overline{g}) = 0$ for all $\overline{g} \neq g$.
    Let the allocation $\widehat{\mathcal{A}}$ in $\mathcal{I}^*$ be the same as $\mathcal{A}$ except that there is an additional bundle $\widehat{A}_0 = \{g\}$. 
    
    If $g \in F_i(\widehat{A}_i)$ for some $i \in N$ or 
    there exists a nondegenerate path $P_{g',g}$ in the exchange graph $\mathcal{G}(\widehat{\mathcal{A}})$ that begins at some good $g' \in F_i(\widehat{A}_i)$ for some $i \in N$ and ends at $g \in \widehat{A}_0$,
	then for exactly one agent $z \in N$, $|A'_z| = |A_z| + 1$, and for all other agents $y \in N \setminus \{z\}$, $|A'_y| = |A_y|$.
\end{lemma}
\begin{proof}
	By definition of \ruletie{}, both $\mathcal{A}$ and $\mathcal{A}'$ are non-redundant.
    Moreover, by the resource-monotonicity of \ruletie{} (Theorem \ref{thm:resource_monotonicity}), $|A'_k| \geq |A_k|$ for all $k \in N$.
    Suppose that $g \in F_i(\widehat{A}_i)$ for some $i \in N$ or 
    there exists a nondegenerate path $P_{g',g}$ in the exchange graph $\mathcal{G}(\widehat{\mathcal{A}})$ that begins at some good $g' \in F_i(\widehat{A}_i)$ for some $i \in N$ and ends at $g \in \widehat{A}_0$.

    First, assume for a contradiction that $|A'_k| = |A_k|$ for all $k \in N$.
    If $g \in F_i(\widehat{A}_i)$ for some $i \in N$, by starting from $\widehat{\mathcal{A}}$ and transferring $g$ from $\widehat{A}_0$ to $\widehat{A}_i$, we obtain a non-redundant allocation in which agent~$i$'s utility increases by~$1$, agent~$0$'s utility decreases by~$1$, and every other agent's utility remains the same.
    Similarly, if $g\not\in F_j(\widehat{A}_j)$ for all $j\in N$ and there exists a nondegenerate path from $F_i(\widehat{A}_i)$ to $\widehat{A}_0$ for some $i\in N$, by considering a shortest such path, performing augmentation along this path, and applying Lemma~\ref{lemma:gsp_schrijver}, we obtain a non-redundant allocation in which agent $i$'s utility increases by~$1$, agent $0$'s utility decreases by~$1$, and every other agent's utility remains the same.
    Hence, in either case, we obtain an allocation that Pareto-dominates $\mathcal{A}'$ in $\mathcal{I}'$, contradicting the fact that $\mathcal{A}'$ is an \ruletie{} allocation in $\mathcal{I}'$.
    We therefore conclude that there exists an agent $j\in N$ such that $|A'_j| > |A_j|$.

    Next, assume for a contradiction that there exists an agent $z \in N$ such that $|A'_z| > |A_z| + 1$.
    Consider the allocation $\widehat{\mathcal{A}}'$ in $\mathcal{I^*}$ that is equivalent to $\mathcal{A}'$, but with the addition that $\widehat{A}'_0 = \emptyset$.
	We also have $|\widehat{A}'_z| > |\widehat{A}_z| + 1$.
    Note that $\widehat{\mathcal{A}}'$ and $\widehat{\mathcal{A}}$ are non-redundant, $\widehat{\mathcal{A}}$ is Pareto-optimal because agent~$0$ only positively values $g$ and $\mathcal{A}$ is Pareto-optimal in $\mathcal{I}$, and $|\widehat{A}'_0| = 0$.
	Invoking Lemma \ref{lemma:resmon_bound_oneagent}, we find that there exists an agent different from agent~$0$ belonging to $\alow(\widehat{\mathcal{A}}',\widehat{\mathcal{A}})$.
	This also means that there exists an agent different from agent~$0$ belonging to $\alow(\mathcal{A}',\mathcal{A})$.
    However, this contradicts the fact that $|A'_k| \ge |A_k|$ for all $k\in N$.
	Hence, we must have $|A_k| \leq |A'_k| \leq |A_k|+1$ for all $k \in N$.

    Finally, assume for a contradiction that there exist two distinct agents $y, z \in N$ such that $|A'_y| = |A_y| + 1$ and $|A'_z| = |A_z| +1$.
	Again, consider the allocation $\widehat{\mathcal{A}}'$ in $\mathcal{I}^*$ that is equivalent to $\mathcal{A}'$, but with the addition that $\widehat{A}'_0 = \emptyset$.
	We also have $|\widehat{A}'_y| = |\widehat{A}_y| + 1$ and $|\widehat{A}'_z| = |\widehat{A}_z| +1$, that is, $y,z \in \ahigh(\widehat{\mathcal{A}}',\widehat{\mathcal{A}})$.
	As in the previous paragraph, invoking Lemma~\ref{lemma:resmon_bound_oneagent}, we find that there exists an agent different from agent~$0$ belonging to $\alow(\widehat{\mathcal{A}}', \widehat{\mathcal{A}})$.
	This also means that there exists an agent different from agent~$0$ belonging to $\alow(\mathcal{A}',\mathcal{A})$.
    However, this contradicts the fact that $|A'_k| \ge |A_k|$ for all $k\in N$.
    We therefore conclude that for exactly one agent $r \in N$, $|A'_r| = |A_r| + 1$, and for all other agents $s \in N \setminus \{r\}$, $|A'_s| = |A_s|$.
\end{proof}

We are now ready to present an algorithm that computes an \ruletie{} allocation, which is shown as Algorithm~\ref{alg:therule}.
The algorithm starts with an empty allocation, and allocates one good at a time by calling the subroutine \texttt{AddOneGood} (Algorithm~\ref{alg:addonegood}) with the current allocation and the new good to be allocated.
Note that \texttt{AddOneGood} may reallocate some goods that have already been allocated.

\begin{algorithm}[tb]
    \caption{\ruletie{} algorithm}
    \label{alg:therule}
    \textbf{Input}: Set of agents $N=\{1,\dots,n\}$, set of goods $G=\{g_1,\dots,g_m\}$, weight vector $\mathbf{w} = (w_1, \dots, w_n)$, and valuation profile $\mathbf{v} = (v_1, \dots, v_n)$ accessible via queries\\
    \textbf{Output}: \ruletie{} allocation of the $m$ goods in $G$
    
    \begin{algorithmic}[1]
        \STATE Initialize the empty allocation $\mathcal{A}^0$ where $A_i^0 = \emptyset$ for all $i \in N$.
        \FOR{$t = 1,2,\dots,m$}
        \STATE $\mathcal{A}^t \leftarrow \texttt{AddOneGood}(\mathcal{A}^{t-1}, g_t)$ (see Algorithm~\ref{alg:addonegood}) \\
        \ENDFOR
        \STATE \textbf{return} $\mathcal{A}^m$
    \end{algorithmic}
\end{algorithm}

\begin{algorithm}[tb]
    \caption{\texttt{AddOneGood}}
    \label{proc:addonegood}
    \textbf{Input}: \ruletie{} allocation $\mathcal{A}^{t-1}$ ($t-1$ goods in total), and good $g$\\
    \textbf{Output}: \ruletie{} allocation $\mathcal{A}^{t}$ ($t$ goods in total)
    
    \begin{algorithmic}[1]
    \STATE Let $\widehat{\mathcal{A}}^{t-1}$ be the allocation that is the same as $\mathcal{A}^{t-1}$, but with an added dummy agent~$0$ such that $\widehat{A}_0^{t-1} = \{g\}$. Assume that agent~$0$ has an arbitrary positive weight and a binary additive valuation function such that $v_0(g) = 1$ and $v_0(\overline{g}) = 0$ for all $\overline{g} \neq g$. 
    \STATE Construct the exchange graph $\mathcal{G}(\widehat{\mathcal{A}}^{t-1})$.
    \IF{$g \notin F_i(\widehat{A}^{t-1}_i)$ for all $i \in N$ and there does not exist a nondegenerate path in $\mathcal{G}(\widehat{\mathcal{A}}^{t-1})$ that starts from a good in $F_j(\widehat{A}^{t-1}_j)$ for some $j\in N$ and ends with $g$}
        \STATE \textbf{return} the allocation $\mathcal{A}^{t-1}$ with $g$ not assigned to any agent
    \ENDIF
    \STATE $\mathcal{P} \leftarrow \emptyset$ 
    \FOR{each agent $i = 1, \dots, n$}
        \IF{$g \in F_i(\widehat{A}_i^{t-1})$}
            \STATE Add the tuple $(i, P_{g,g}, \mathbf{u}_i)$ to $\mathcal{P}$, where $\mathbf{u}_i$ is the  utility vector of the allocation $\widehat{\mathcal{A}}^{t-1}$ with $i$'s utility increased by~$1$, and $P_{g,g}$ is the degenerate path from $g$ to itself.
        \ENDIF
        \IF{a nondegenerate path from some $g' \in F_i(\widehat{A}_i^{t-1})$ to $g \in \widehat{A}_0^{t-1}$ exists}   
            \STATE Let $P_{g',g}$ be a shortest such path.
            \STATE Add the tuple $(i, P_{g',g}, \mathbf{u}_i)$ to $\mathcal{P}$, where $\mathbf{u}_i$ is the  utility vector of the allocation $\widehat{\mathcal{A}}^{t-1}$ with $i$'s utility increased by~$1$. 
        \ENDIF
    \ENDFOR 
    \STATE Select the tuple $Q \in \mathcal{P}$ with the maximum weighted welfare according to \therule{} across all tuples (computed using $\mathbf{u}_i$), breaking ties according to Definitions~\ref{def:Rf} and \ref{def:ruletie}. Let $j$ and $P$ be the agent and the path in the tuple~$Q$, respectively. \label{line:select} 
    \STATE Let $\mathcal{A}^t$ be the allocation where $A_i^t = A_i^{t-1} \triangle P$ for all $i \in N$ (see Section~\ref{sec:exchange} for the definition of~$\triangle$), and add the first good on the path~$P$ to $A_j^t$.
    \STATE \textbf{return} $\mathcal{A}^t$
    \end{algorithmic}
\label{alg:addonegood}
\end{algorithm}

\begin{theorem}
	Under matroid-rank valuations, for any concave function $f$, Algorithm \ref{alg:therule} computes an \ruletie{} allocation in polynomial time.
\end{theorem}

\begin{proof}
	To prove that the allocation returned by the algorithm is an \ruletie{} allocation, we show that after every iteration, the allocation returned by the subroutine \texttt{AddOneGood} is \ruletie{} with respect to the goods that are already allocated.
	
	Let $t\in\{1,\dots,m\}$ be some iteration of the \textbf{for} loop of Algorithm~\ref{alg:therule}, and let $g$ be the good added during this iteration.
	As in Algorithm~\ref{alg:addonegood}, let $\widehat{\mathcal{A}}^{t-1}$ denote the allocation that is the same as $\mathcal{A}^{t-1}$, but with an additional agent~$0$ with bundle $\widehat{A}_0^{t-1}=\{g\}$.
    Assume that agent~$0$ has an arbitrary positive weight and a binary additive valuation function such that $v_0(g) = 1$ and $v_0(\overline{g}) = 0$ for all $\overline{g} \neq g$. 
	
	Consider the exchange graph $\mathcal{G}(\widehat{\mathcal{A}}^{t-1})$.
	If $g \notin F_i(\widehat{A}_i^{t-1})$ for all $i \in N$ and there does not exist a nondegenerate path in $\mathcal{G}(\widehat{\mathcal{A}}^{t-1})$ that starts from a good in $F_j(\widehat{A}^{t-1}_j)$ for some $j\in N$ and ends with~$g$, then by Lemma \ref{lem:correctness_nopath}, the allocation $\mathcal{A}^t = \mathcal{A}^{t-1}$ (with $g$ not assigned to any agent) is \ruletie{} for the first $t$ goods, and \texttt{AddOneGood} correctly returns this allocation.
    Suppose from now on that either $g \in F_i(\widehat{A}_i^{t-1})$ for some $i \in N$, or there is a nondegenerate path $P_{g',g}$ that begins at some good $g' \in F_i(\widehat{A}^{t-1}_i)$ for some $i \in N$ and ends at $g \in \widehat{A}^{t-1}_0$.
	This means that at least one tuple is added to $\mathcal{P}$ and this iteration of \texttt{AddOneGood} terminates with an allocation $\mathcal{A}^t$.
	
	Let $\widetilde{\mathcal{A}}^t$ be an \ruletie{} allocation of the $t$ goods.
	By Lemma \ref{lem:correctness_gotpath}, exactly one agent $z\in N$ has $|\widetilde{A}^t_z|=|A^{t-1}_z|+1$, while all other agents $y \in N\setminus\{z\}$ have $|\widetilde{A}_y^t|=|A_y^{t-1}|$.
	Let $\overline{\mathcal{A}}^t$ be the allocation that is the same as $\widetilde{\mathcal{A}}^t$, but with agent~$0$ with bundle $\overline{A}^t_0 = \emptyset$. 
    Then, we have that $\alow(\widehat{\mathcal{A}}^{t-1},\overline{\mathcal{A}}^t) = \{z\}$ and $\ahigh(\widehat{\mathcal{A}}^{t-1},\overline{\mathcal{A}}^t) = \{0\}$. 
    We consider two cases.

    \begin{description}
        \item[Case 1: $g \in F_z(\widehat{A}_z^{t-1})$.]
        Then the tuple $(z, P_{g,g} , \mathbf{u}_{z})$ is added to $\mathcal{P}$.
        By adding $g$ to $A_z^{t-1}$, we obtain an allocation in which agent~$z$'s utility increases by~$1$ and the utility of every other agent in~$N$ remains the same.
        \item[Case 2: $g \notin F_z(\widehat{A}_z^{t-1})$.]        
        Note that $\widehat{\mathcal{A}}^{t-1}$ and $\overline{\mathcal{A}}^t$ are non-redundant.
        Moreover, $\widetilde{\mathcal{A}}^t$ Pareto-dominates ${\mathcal{A}}^{t-1}$, which means that no allocation among the agents in $N$ that leaves $g$ unallocated can be an \ruletie{} allocation of the $t$ goods.
        This implies that $\overline{\mathcal{A}}^t$ is Pareto-optimal.
        Since $\alow(\widehat{\mathcal{A}}^{t-1},\overline{\mathcal{A}}^t) = \{z\}$, $\ahigh(\widehat{\mathcal{A}}^{t-1},\overline{\mathcal{A}}^t) = \{0\}$, and $\widehat{A}_0^{t-1}=\{g\}$, by Lemma \ref{lemma:gsp_BarmanVerma_path},  there exists a path
    in $\mathcal{G}(\widehat{A}^{t-1})$ from some good $g' \in F_z(\widehat{A}^{t-1}_z)$ to the good $g \in \widehat{A}^{t-1}_0$. 
    Since $g \notin F_z(\widehat{A}_z^{t-1})$, this path must be nondegenerate.
    Hence, a tuple $(z, P_{g'',g} , \mathbf{u}_{z})$ corresponding to a shortest path is added to $\mathcal{P}$.
    According to Lemma~\ref{lemma:gsp_schrijver}, by augmenting along this path, we obtain an allocation in which agent~$z$'s utility increases by~$1$ and the utility of every other agent in~$N$ remains the same.
    
    \end{description}
	In both cases, since the utility vector $\mathbf{u}_{z}$ is the same as the utility vector of $\widetilde{\mathcal{A}}^t$, and by the selection method in line~\ref{line:select} of Algorithm \ref{alg:addonegood} along with the fact that every tuple added to $\mathcal{P}$ corresponds to a valid allocation (which can be shown by a similar reasoning as in the two cases), the allocation $\mathcal{A}^t$ returned by \texttt{AddOneGood} is also an \ruletie{} allocation. 
	This completes the proof of correctness.
	
	Finally, we analyze the running time of Algorithm~\ref{alg:therule}. 
	In the main algorithm, there are $O(m)$ iterations of the \textbf{for} loop. In the \texttt{AddOneGood} subroutine, constructing the graph $\mathcal{G}(\widehat{A}^{t-1})$ takes $O(m^2)$ time, checking whether $g\not\in F_i(\widehat{A}_i^{t-1})$ for all~$i\in N$ takes $O(n)$ time, checking whether there exists a nondegenerate path in $\mathcal{G}(\widehat{\mathcal{A}}^{t-1})$ that starts from a good in $F_j(\widehat{A}^{t-1}_j)$ for some $j\in N$ and ends with~$g$ takes $O(m^2n)$ time using breadth-first search backward from $g$ \citep{CormenLeRi09}, and there are $O(n)$ iterations of the \textbf{for} loop. Within each iteration, determining whether a specified path exists and, if so, finding a shortest path can be done using breadth-first search backward from $g$, which takes $O(m^2)$ time. In the last step, selecting the tuple in $\mathcal{P}$ takes $O(n^2)$ time and computing the new allocation takes $O(m)$ time. Putting everything together, Algorithm \ref{alg:therule} terminates in $O(mn(m^2+n))$ time.
\end{proof}

\section{Population- and Weight-Monotonicity}

In this section, we address population- and weight-monotonicity.
Population-monotonicity states that if an extra agent is added, then no original agent's utility should increase.

\begin{definition}[Population-monotonicity]
    An allocation rule $\mathcal{F}$ is \emph{population-monotone} if the following holds: For any two instances $\mathcal{I}$ and $\mathcal{I}'$ such that $\mathcal{I}'$ can be obtained from $\mathcal{I}$ by adding one extra agent, if $\mathcal{F}(\mathcal{I}) = \mathcal{A}$ and $\mathcal{F}(\mathcal{I}') = \mathcal{A}'$, then for each agent $i$ in the original set of agents, $v_i(A_i) \geq v_i(A'_i)$.    
\end{definition}

As with resource-monotonicity, under general additive valuations in the unweighted setting, MNW fails to be population-monotone regardless of tie-breaking \citep{ChakrabortyScSu21}.
By contrast, in the weighted setting, MWNW with lexicographical tie-breaking satisfies population-monotonicity under binary additive valuations \citep{SuksompongTe22}.

We will leverage the subroutine \texttt{AddOneGood} from Section~\ref{sec:algo} to prove that \ruletie{} is population-monotone for any concave function~$f$.
Recall from Section~\ref{sec:prelim} that for an allocation $\mathcal{A}$ and a subset of agents $S\subseteq N$,  $\mathcal{A}_S$ denotes the allocation derived from restricting $\mathcal{A}$ to the bundles of the agents in $S$ (in the same order as in $\mathcal{A})$, and $N^+_\mathcal{A}\subseteq N$ denotes the subset of agents receiving positive utility from~$\mathcal{A}$.

\begin{lemma}\label{lemma:subset-therule}
	Suppose we have an \ruletie{} allocation $\mathcal{A} = (A_1,\dots,A_n)$ for the agents in $N$. 
	Then, for any subset of agents $S \subseteq N$, $\mathcal{A}_S$ is also an \ruletie{} allocation in the instance with the set of agents $S$ and the set of goods in $\mathcal{A}_S$.
\end{lemma}

\begin{proof}
	Suppose, for a contradiction, that $\mathcal{A}_S$ is not an \ruletie{} allocation. 
	This means that there exists another allocation $\mathcal{A}'_S$ of the goods in $\mathcal{A}_S$
    to the agents in $S$ that is preferred by the rule~\ruletie{}.
    For notational consistency, let $\mathcal{A}'$ be the allocation in which the agents in $N\setminus S$ are allocated according to~$\mathcal{A}$ while the agents in $S$ are allocated according to $\mathcal{A}'_S$.

    Assume first that $f(0) \ne -\infty$.
    We have that either
    \begin{enumerate}[label=(\roman*)]
    \item $\sum_{i \in S} w_i \cdot f(v_i(A_i))< \sum_{i \in S} w_i \cdot f(v_i(A'_i))$, or
    \item $\sum_{i \in S} w_i \cdot f(v_i(A_i))= \sum_{i \in S} w_i \cdot f(v_i(A'_i))$, but the utility vector of $\mathcal{A}'_S$ is lexicographically preferred to that of $\mathcal{A}_S$.
    \end{enumerate}
    In case~(i), since the agents in $N\setminus S$ are allocated in the same way in $\mathcal{A}'$ as in $\mathcal{A}$, we have $\sum_{i \in N} w_i \cdot f(v_i(A_i))< \sum_{i \in N} w_i \cdot f(v_i(A'_i))$, contradicting the fact that $\mathcal{A}$ is an \ruletie{} allocation for the agents in~$N$.
    In case~(ii), we have that $\sum_{i \in N} w_i \cdot f(v_i(A_i)) = \sum_{i \in N} w_i \cdot f(v_i(A'_i))$ but the utility vector of $\mathcal{A}'$ is lexicographically preferred to that of $\mathcal{A}$, again a contradiction.
   
    Next, assume that $f(0) = -\infty$.
    We have that
	\begin{enumerate}[label=(\roman*)]
		\item $|N^+_{\mathcal{A}_S}| < |N^+_{\mathcal{A}_S'}|$, or
		\item $|N^+_{\mathcal{A}_S}| = |N^+_{\mathcal{A}'_S}|$, but the indices of the agents in $N^+_{\mathcal{A}'_S}$ are lexicographically lower than those in $N^+_{\mathcal{A}_S}$, or
        \item $N^+_{\mathcal{A}_S} = N^+_{\mathcal{A}'_S}$, but $\sum_{i \in N^+_{\mathcal{A}_S}} w_i \cdot f(v_i(A_i))< \sum_{i \in N^+_{\mathcal{A}'_S}} w_i \cdot f(v_i(A'_i))$, or 
		\item $N^+_{\mathcal{A}_S} = N^+_{\mathcal{A}'_S}$ and $\sum_{i \in N^+_{\mathcal{A}_S}} w_i \cdot f(v_i(A_i)) = \sum_{i \in N^+_{\mathcal{A}'_S}} w_i \cdot f(v_i(A'_i))$, but the  utility vector of $\mathcal{A}'_S$ is lexicographically preferred to that of $\mathcal{A}_S$.
	\end{enumerate}	
	Case (i) means that the number of agents in $N^+_{\mathcal{A}}$ is not maximized, a contradiction.
	In case (ii), we have that $|N^+_{\mathcal{A}}| = |N^+_{\mathcal{A}'}|$ but the indices of the agents in $N^+_{\mathcal{A}'}$ are lexicographically lower than those in $N^+_{\mathcal{A}}$, a contradiction.
    In case~(iii), we have that $N^+_{\mathcal{A}} = N^+_{\mathcal{A}'}$ but $\sum_{i \in N^+_{\mathcal{A}}} w_i \cdot f(v_i(A_i))< \sum_{i \in N^+_{\mathcal{A}'}} w_i \cdot f(v_i(A'_i))$, again a contradiction.
    Finally, in case~(iv), we have that $N^+_{\mathcal{A}} = N^+_{\mathcal{A}'}$ and $\sum_{i \in N^+_{\mathcal{A}}} w_i \cdot f(v_i(A_i)) = \sum_{i \in N^+_{\mathcal{A}'}} w_i \cdot f(v_i(A'_i))$ but the utility vector of $\mathcal{A}'$ is lexicographically preferred to that of $\mathcal{A}$, yielding yet another contradiction.
\end{proof}

We remark that Lemma~\ref{lemma:subset-therule} holds even for arbitrary monotone and normalized valuations, provided that the domain of the function $f$ is extended from $\mathbb{Z}_{\ge 0}$ to $\mathbb{R}_{\ge 0}$.
Indeed, the same proof still works in that more general case.

\begin{theorem}\label{thm-populationmonotone}
    Under matroid-rank valuations, the rule \emph{\ruletie{}} with any concave function $f$ is population-monotone.
\end{theorem}
\begin{proof}
    Assume that $f$ is concave.
    Consider an \ruletie{} allocation $\mathcal{A}$ for $n+1$ agents. Let $S$ be any subset of $n$ agents, and let $i$ be the agent not in $S$. 
    By Lemma~\ref{lemma:subset-therule}, $\mathcal{A}_S$ is an \ruletie{} allocation for the agents in~$S$. 
    Then, with the bundle of goods remaining to be allocated being that of agent $i$, iteratively use the subroutine \texttt{AddOneGood} to allocate all of these goods, giving us an \ruletie{} allocation of the original set of goods to the agents in~$S$. 
    Since \texttt{AddOneGood} never decreases the size of any agent's bundle, which is equal to the agent's utility (due to non-redundancy), population-monotonicity is satisfied.
\end{proof} 

The final monotonicity property that we consider is weight-monotonicity, which states that if the weight of some agent increases, then the agent's utility should not decrease.

\begin{definition}[Weight-monotonicity]
    An allocation rule $\mathcal{F}$ is \emph{weight-monotone} if the following holds: For any two instances $\mathcal{I}$ and $\mathcal{I'}$ such that $\mathcal{I}'$ can be obtained from $\mathcal{I}$ by increasing the weight of an agent $i \in N$, if $\mathcal{F}(\mathcal{I}) = \mathcal{A}$ and $\mathcal{F}(\mathcal{I}') = \mathcal{A}'$, then $v_i(A_i) \leq v_i(A'_i)$.
\end{definition}

We show next that \ruletie{} satisfies weight-monotonicity.
This result does not require $f$ to be concave, but we still assume as in Definition~\ref{def:Rf} that $f$ is strictly increasing.

\begin{theorem}
\label{thm:weight-mon}
	Under matroid-rank valuations, the rule \emph{\ruletie{}} with any function $f$ is weight-monotone.
\end{theorem}

\begin{proof}
    Fix an instance $\mathcal{I}$ with a set of agents~$N$ and $\mathbf{w} = (w_1,\dots,w_n)$.
    Consider a modified instance~$\mathcal{I}'$ such that $w'_i > w_i$ for some $i \in N$ and $w'_k = w_k$ for all $k \in N \setminus \{i\}$.
    Let \ruletie{}$(\mathcal{I}) = \mathcal{A}$ and \ruletie{}$(\mathcal{I}')= \mathcal{A}'$.
    Suppose for a contradiction that $v_i(A_i) > v_i(A'_i)$.

    Assume first that $\sum_{j\in N} w_j\cdot f(v_j(A_j)) \ne -\infty$.
    Since the valuation functions in $\mathcal{I}$ and $\mathcal{I}'$ are the same, we must also have $\sum_{j\in N} w'_j\cdot f(v_j(A'_j)) \ne -\infty$.
    Because $v_i(A_i) > v_i(A'_i)$ and $f$ is strictly increasing, we have that 
    \begin{equation} \label{eqn:weightmon_increasingcond}
        f(v_i(A_i)) > f(v_i(A'_i)).
    \end{equation}
    Also, $w'_i - w_i > 0$. Multiplying $w'_i - w_i$ on both sides of (\ref{eqn:weightmon_increasingcond}), we get
    \begin{equation} \label{eqn:weight_monotonicity_difference}
        (w'_i - w_i) \cdot f(v_i(A_i)) > (w'_i - w_i) \cdot f(v_i(A'_i)).
    \end{equation}
    Note that
    \begin{align*}
        &(w'_i - w_i) \cdot f(v_i(A_i)) \\
        &= \left( w'_i \cdot f(v_i(A_i)) + \sum_{k\in N \setminus \{i\}} w_k \cdot f(v_k(A_k)) \right) - \left( w_i \cdot f(v_i(A_i)) + \sum_{k\in N \setminus \{i\}} w_k \cdot f(v_k(A_k)) \right)
    \end{align*}
    and
    \begin{align*}
        &(w'_i - w_i) \cdot f(v_i(A'_i))\\
        &= \left( w'_i \cdot f(v_i(A'_i)) + \sum_{k\in N \setminus \{i\}} w_k \cdot f(v_k(A'_k)) \right) - \left( w_i \cdot f(v_i(A'_i)) + \sum_{k\in N \setminus \{i\}} w_k \cdot f(v_k(A'_k)) \right).
    \end{align*}
    Using (\ref{eqn:weight_monotonicity_difference}) and rearranging, we get
    \begin{align*}
        &\left( w'_i \cdot f(v_i(A_i)) + \sum_{k\in N \setminus \{i\}} w_k \cdot f(v_k(A_k)) \right)
        - \left(w'_i \cdot f(v_i(A'_i)) + \sum_{k\in N \setminus \{i\}} w_k \cdot f(v_k(A'_k)) \right) \\
        &> \left( w_i \cdot f(v_i(A_i)) + \sum_{k\in N \setminus \{i\}} w_k \cdot f(v_k(A_k)) \right) - \left( w_i \cdot f(v_i(A'_i)) + \sum_{k\in N \setminus \{i\}} w_k \cdot f(v_k(A'_k)) \right) 
        \ge 0,
    \end{align*}
    where the latter inequality holds because $\mathcal{A}$ is an \ruletie{} allocation in $\mathcal{I}$ and is preferred to $\mathcal{A}'$ under~$\mathcal{I}$.
    This means that in $\mathcal{I}'$, the allocation $\mathcal{A}$ is preferred to $\mathcal{A}'$, a contradiction with \ruletie{}$(\mathcal{I}')= \mathcal{A}'$.

    Next, assume that $\sum_{j\in N} w_j\cdot f(v_j(A_j)) = -\infty$.
    Due to the tie-breaking in Definition~\ref{def:Rf}, and since the valuation functions in $\mathcal{I}$ and $\mathcal{I}'$ are the same, we have $N^+_{\mathcal{A}} = N^+_{\mathcal{A}'}$.
    Since $v_i(A_i) > v_i(A'_i)$, both $v_i(A_i)$ and $v_i(A'_i)$ are positive.
    We can then apply a similar argument as above with $N$ replaced by $N^+_{\mathcal{A}}$.
\end{proof}

Similarly to Lemma~\ref{lemma:subset-therule}, Theorem~\ref{thm:weight-mon} holds even for arbitrary monotone and normalized valuations, provided that the domain of the function $f$ is extended from $\mathbb{Z}_{\ge 0}$ to $\mathbb{R}_{\ge 0}$.
Indeed, the same proof still works in that more general case.

\section{Group-Strategyproofness}

In this section, we turn our attention to strategyproofness, an important property which states that no agent can strictly benefit by misreporting her preferences. 
A more robust version of strategyproofness is group-strategyproofness, which guarantees that no group of agents can misreport their preferences in such a way that every agent in the group strictly benefits.

\begin{definition}[Group-strategyproofness]
	An allocation rule $\mathcal{F}$ satisfies \emph{group-strategyproofness} if  there do not exist valuation profiles $\mathbf{v}$ and $\mathbf{v}'$ and a nonempty set of agents $S \subseteq N$ such that, if we denote the instances corresponding to $\mathbf{v}$ and $\mathbf{v}'$ by $\mathcal{I}$ and $\mathcal{I}'$, respectively, then
	\begin{itemize}
		\item $v_i(A'_i) > v_i(A_i)$ for all $i \in S$, and
		\item $v_j = v'_j$ for all $j \in N \setminus S$,
	\end{itemize}
    where $\mathcal{A} = \mathcal{F}(\mathcal{I})$ and $\mathcal{A}' = \mathcal{F}(\mathcal{I}')$.
\end{definition}

Prior work has shown that MWNW with lexicographical tie-breaking satisfies group-strategyproofness under binary additive valuations \citep{SuksompongTe22}. 
We strengthen this result by showing that in the matroid-rank setting, a more general class of rules provides the same guarantee.
We also remark that one could consider a stronger version of strategyproofness, where no group of agents can misreport their preferences in such a way that every agent in the group weakly benefits and at least one agent strictly benefits; however, \citet[Prop.~1]{SuksompongTe22} proved that even in the unweighted setting and under binary additive valuations, MNW with lexicographical tie-breaking fails this stronger version.

\begin{theorem}
    Under matroid-rank valuations, the rule \ruletie{} with any concave function $f$ is group-strategyproof.
\end{theorem}
\begin{proof}
    Assume that $f$ is concave.
    Suppose for a contradiction that there exist valuation profiles $\mathbf{v}$ and~$\mathbf{v}'$ and a nonempty set of agents $S \subseteq N$ such that 
    \begin{enumerate}[label=(\roman*)]
        \item $v_i(A_i^\lie) > v_i(A_i^\truth)$ for all $i \in S$, and
        \item $v_j = v'_j$ for all $j \in N \setminus S$,
    \end{enumerate}
    where $\mathcal{A}^\truth$ is an \ruletie{} allocation under the truthful valuation profile $\mathbf{v}$, and $\mathcal{A}^\lie$ is an \ruletie{} allocation under the valuation profile $\mathbf{v}'$, which coincides with $\mathbf{v}$ except possibly for agents in $S$.
    Let $\widetilde{\mathcal{A}}^\lie$ be a non-redundant allocation with respect to $\mathbf{v}$ such that $\widetilde{A}_k^\lie \subseteq A_k^\lie$ and $v_k(\widetilde{A}_k^\lie) = v_k(A_k^\lie)$ for each $k \in N$.
    This means that
    \begin{equation} \label{eqn:gsp_concave_nonredundantInLie}
        |A_k^\lie| \geq v_k(A_k^\lie) = v_k(\widetilde{A}_k^\lie) = |\widetilde{A}_k^\lie|
    \end{equation} 
    for all $k\in N$.
    Note that $\mathcal{A}^\truth$ is non-redundant with respect to $\mathbf{v}$ by definition of \ruletie{}.
    Combining this fact together with (i), we have that for all $i \in S$, 
    \begin{equation} \label{eqn:gsp_concave_nonredundant}
        |\widetilde{A}_i^\lie| = v_i(\widetilde{A}_i^\lie) = v_i(A_i^\lie) > v_i(A_i^\truth) = |A_i^\truth|.
    \end{equation}
    In other words, $S\subseteq \ahigh(\widetilde{\mathcal{A}}^\lie,\mathcal{A}^\truth)$.
    Let $C$ be the set containing all agents $i\in S$ who maximize the expression
    \begin{equation*}
         w_i \cdot f(|A_i^\lie|) - w_i \cdot f(|A_i^\lie| - 1).
    \end{equation*}
    among all agents in $S$.
    Note that the expression above is valid because, by (\ref{eqn:gsp_concave_nonredundantInLie}) and (\ref{eqn:gsp_concave_nonredundant}), $|A_i^\lie| \geq 1$ for all $i\in S$.\footnote{It is possible that $w_i \cdot f(|A_i^\lie|) - w_i \cdot f(|A_i^\lie| - 1) = \infty$, if $|A_i^\lie| = 1$ and $f(0) = -\infty$.}
    Let $h$ denote the agent with the smallest index among all agents in $C$; in particular, $h\in S$.
    Since $S\subseteq \ahigh(\widetilde{\mathcal{A}}^\lie,\mathcal{A}^\truth)$, we also have $h\in \ahigh(\widetilde{\mathcal{A}}^\lie,\mathcal{A}^\truth)$.
    The remainder of this proof will take the perspective of $\mathbf{v}$ unless otherwise stated.

    Observe that $\widetilde{\mathcal{A}}^\lie$ is non-redundant and $\mathcal{A}^\truth$ is non-redundant and Pareto-optimal (by definition of an \ruletie{} allocation).
    Thus, by Lemma \ref{lemma:gsp_BarmanVerma_path}, for the agent $h\in \ahigh(\widetilde{\mathcal{A}}^\lie,\mathcal{A}^\truth)$, there exists an agent $\ell \in \alow(\widetilde{\mathcal{A}}^\lie, \mathcal{A}^\truth)$ along with a simple directed path $P = (g_t,\dots,g_1)$ in the exchange graph $\mathcal{G}(\widetilde{\mathcal{A}}^\lie)$ such that the source vertex is $g_t \in A_\ell^\truth \cap F_\ell(\widetilde{A}_\ell^\lie)$ and the sink vertex is $g_1 \in \widetilde{A}_h^\lie$.

    Let $Q$ be a shortest path from $F_\ell(\widetilde{A}_\ell^\lie)$ to $\bigcup_{i \in S} \widetilde{A}_i^\lie$; the existence of $P$ guarantees that such a path exists.
    Suppose that $Q$ ends at $\widetilde{A}_b^\lie$ for some $b \in S$ (possibly $b = h$). 
    Then, for the path $Q$ in $\mathcal{G}(\widetilde{\mathcal{A}}^\lie)$, only its sink vertex lies in $\bigcup_{i \in S} \widetilde{A}_i^\lie$.
    In particular, all owners of goods on $Q$ besides $b$ belong to $N \setminus S$.
    
    Now, for every agent $j \in N \setminus S$, we have that $v_j = v'_j$ and so $\widetilde{A}_j^\lie = A_j^\lie$.
    This means that the path~$Q$ also exists in the exchange graph $\mathcal{G}(\mathcal{A}^\lie)$, where we define $\mathcal{G}(\mathcal{A}^\lie)$ with respect to $\mathbf{v}'$ (rather than to $\mathbf{v}$).  
    Note also that $\widetilde{A}_b^\lie \subseteq A_b^\lie$.
    Under $\mathbf{v}'$, since $\mathcal{A}^\lie$ is non-redundant, by considering a shortest path from $F_\ell(A_\ell^\lie)$ to $A_b^\lie$ and applying Lemma~\ref{lemma:gsp_schrijver}, we obtain an allocation $\mathcal{X}^\lie$ that is non-redundant with respect to $\mathbf{v}'$ such that  
    \begin{equation*}
        |X_\ell^\lie| = |A_\ell^\lie|+1, |X_b^\lie| = |A_b^\lie|-1, \text{ and } |X_k^\lie| = |A_k^\lie| \text{ for all } k \in N \setminus \{\ell,b\}.
    \end{equation*}
    Also, by our choice of $h$, we have
    \begin{equation} \label{eqn:gsp_hb}
        w_b \cdot f(|A_b^\lie|) - w_b \cdot f(|A_b^\lie|-1) \leq w_h \cdot f(|A_h^\lie|) - w_h \cdot f(|A_h^\lie|-1),
    \end{equation}
    where equality holds only if $h\le b$.
    
    Next, the second condition of Lemma~\ref{lemma:gsp_BarmanVerma_path} implies that, from $\mathcal{A}^\truth$, we can obtain an allocation $\mathcal{X}^\truth$ that is non-redundant with respect to $\mathbf{v}$ such that
    \begin{equation*}
        |X_\ell^\truth| = |A_\ell^\truth| - 1, |X_h^\truth| = |A_h^\truth| + 1, \text{ and } |X_k^\truth| = |A_k^\truth| \text{ for all } k \in N \setminus \{\ell,h\}.
    \end{equation*}
    Moreover, since $\ell \in \alow(\widetilde{\mathcal{A}}^\lie,\mathcal{A}^\truth)$ and $S\subseteq \ahigh(\widetilde{\mathcal{A}}^\lie,\mathcal{A}^\truth)$, we have $\ell\in N\setminus S$.
    This means that $\widetilde{A}_\ell^\lie = A_\ell^\lie$, and so
    \begin{equation}\label{eqn:gsp_tiebreak_L}
        |A_\ell^\lie| = |\widetilde{A}_\ell^\lie| < |A_\ell^\truth|.
    \end{equation}
    Also, since $b \in S \subseteq \ahigh(\widetilde{\mathcal{A}}^\lie,\mathcal{A}^\truth)$,
    \begin{equation}\label{eqn:gsp_tiebreak_B}
        |A_b^\lie| \geq |\widetilde{A}_b^\lie| > |A_b^\truth|.
    \end{equation}
    Similarly, since $h \in S \subseteq  \ahigh(\widetilde{\mathcal{A}}^\lie,\mathcal{A}^\truth)$, 
    \begin{equation}\label{eqn:gsp_tiebreak_H}
        |A_h^\lie| \geq |\widetilde{A}_h^\lie| > |A_h^\truth|.
    \end{equation}

    Assume first that the function $f$ satisfies $f(0) = -\infty$. We prove five claims. 
    \begin{description}
        \item[Claim 1: $|A^\truth_h| \geq 1$.] 
        Assume for a contradiction that $|A^\truth_h| = 0$.
        We know from (\ref{eqn:gsp_tiebreak_L}) that $|A^\truth_\ell| \geq 1$. If $|A^\truth_\ell| > 1$, then the allocation $\mathcal{X}^\truth$ with utility vector\footnote{Even though we write agent~$\ell$'s utility before agent~$h$'s in the utility vector, the actual order would be reversed if $h < \ell$.} $(|A^\truth_\ell| - 1, |A^\truth_h| + 1)$ for agents $\ell$ and $h$ has strictly more agents receiving positive utility under $\mathbf{v}$ than $\mathcal{A}^\truth$, thereby contradicting the fact that $\mathcal{A}^\truth$ is an \ruletie{} allocation in $\mathbf{v}$.
        Thus, $|A^\truth_\ell| = 1$, which by (\ref{eqn:gsp_tiebreak_L}) means $|A^\lie_\ell| = 0$.
        Since the utility vector $(|A^\truth_\ell|, |A^\truth_h|)$ is preferred to $(|A^\truth_\ell|-1, |A^\truth_h|+1)$ because $\mathcal{A}^\truth$ is an \ruletie{} allocation, it must be that $\ell < h$.
        
        On the other hand, we know from (\ref{eqn:gsp_tiebreak_B}) that $|A^\lie_b| \geq |\widetilde{A}^\lie_b| \geq 1$. 
        Recall from the previous paragraph that $|A^\lie_\ell| = 0$.
        If $|A^\lie_b| > 1$, then the allocation $\mathcal{X}^\lie$ with  utility vector $(|A^\lie_\ell| + 1, |A^\lie_b| - 1)$ for agents $\ell$ and $b$ has strictly more agents receiving positive utility under $\mathbf{v}'$ than $\mathcal{A}^\lie$, thereby contradicting the fact that $\mathcal{A}^\lie$ is an \ruletie{} allocation in $\mathbf{v}'$. 
        Thus, $|A^\lie_b| = 1$.
        Since the utility vector $(|A^\lie_\ell|, |A^\lie_b|)$ is preferred to $(|A^\lie_\ell| + 1, |A^\lie_b| - 1)$ because $\mathcal{A}^\lie$ is an \ruletie{} allocation under $\mathbf{v}'$, we must have that $b < \ell$.
        Combined with the previous paragraph, this gives us $b < h$.
        However, since $|A^\lie_b| = 1$ and $f(0) = -\infty$, this contradicts the definition of $h$.
        Hence, it must be that $|A^\truth_h| \geq 1$.

        \item[Claim 2: $|A^\lie_h| \geq 2$.]
        This follows immediately from (\ref{eqn:gsp_tiebreak_H}) and Claim 1.

        \item[Claim 3: $|A^\lie_b| \geq 2$.]
        By Claim 2, $f(|A_h^\lie|) \neq -\infty$ and $f(|A_h^\lie|-1) \neq -\infty$.
        We know from (\ref{eqn:gsp_tiebreak_B}) that $|A_b^\lie| \geq 1$.
        Assume for a contradiction that $|A_b^\lie| = 1$.
        Then, $f(|A_b^\lie|) \neq -\infty$ and $f(|A_b^\lie|-1) = f(0)= -\infty$, which contradicts the definition of $h$.
        Hence, $|A^\lie_b| \geq 2$.
        
        \item[Claim 4: $|A^\lie_\ell| \geq 1$.]
        Assume for a contradiction that $|A^\lie_\ell| = 0$.
        By Claim 3, the allocation  $\mathcal{X}^\lie$ with utility vector $(|A^\lie_\ell| + 1, |A^\lie_b|-1)$ for agents $\ell$ and $b$ has strictly more agents receiving positive utility than $\mathcal{A}^\lie$ under $\mathbf{v}'$, thereby contradicting the fact that $\mathcal{A}^\lie$ is an \ruletie{} allocation in $\mathbf{v}'$.
        Hence, it must be that $|A^\lie_\ell| \geq 1$.

        \item[Claim 5: $|A^\truth_\ell| \geq 2$.]
        This follows immediately from (\ref{eqn:gsp_tiebreak_L}) and Claim 4.
    \end{description}
    Thus, we can assume henceforth that either $f(0) \neq -\infty$, or if $f(0) = -\infty$, then
    \begin{equation*}
        |A^\truth_\ell| \geq 2, \quad |A^\truth_h| \geq 1, \quad |A^\lie_\ell| \geq 1, \quad
        |A^\lie_h| \geq 2, \quad
        |A^\lie_b| \geq 2.
    \end{equation*}
        
    Under $\mathbf{v}'$, since $\mathcal{A}^\lie$ is an \ruletie{} allocation, and the rule \ruletie{} could have chosen the allocation $\mathcal{X}^\lie$, it must be that 
    \begin{equation} \label{gsp_concave_lie}
    \begin{split}
        w_\ell \cdot f(|A_\ell^\lie|) + w_b \cdot f(|A_b^\lie|) & \geq w_\ell \cdot f(|X_\ell^\lie|) + w_b \cdot f(|X_b^\lie|) \\
        & = w_\ell \cdot f(|A_\ell^\lie|+1) + w_b \cdot f(|A_b^\lie|-1),
    \end{split}
    \end{equation}
    where equality holds only if $b < \ell$.
    Note that by our claims, if $f(0) = -\infty$, then none of the terms $|A_\ell^\lie|, |A_b^\lie|, |A_\ell^\lie|+1, |A_b^\lie|-1$ can be $0$.
    
    Similarly, under $\mathbf{v}$, since $\mathcal{A}^\truth$ is an \ruletie{} allocation, and the rule \ruletie{} could have chosen the allocation $\mathcal{X}^\truth$, it must be that
    \begin{equation} \label{gsp_concave_truth}
    \begin{split}
        w_\ell \cdot f(|A_\ell^\truth|) + w_h \cdot f(|A_h^\truth|) 
        & \geq w_\ell \cdot f(|X_\ell^\truth|) + w_h \cdot f(|X_h^\truth|) \\
        & = w_\ell \cdot f(|A_\ell^\truth|-1) + w_h \cdot f(|A_h^\truth|+1),
    \end{split}
    \end{equation}
    where equality holds only if $\ell < h$.
    Note that by our claims, if $f(0) = -\infty$, then none of the terms $|A_\ell^\truth|, |A_h^\truth|, |A_\ell^\truth|-1, |A_h^\truth|+1$ can be $0$.
    
    Now, by (\ref{eqn:gsp_tiebreak_L}) and the concavity of $f$, we have
    \begin{equation*}
    f(|A_\ell^\lie|+1) - f(|A_\ell^\lie|) \ge f(|A_\ell^\truth|) - f(|A_\ell^\truth| - 1),
    \end{equation*}
    which means that
    \begin{equation} \label{gsp_concave_1}
        f(|A_\ell^\truth| - 1) - f(|A_\ell^\truth|) \geq f(|A_\ell^\lie|) - f(|A_\ell^\lie|+1).
    \end{equation}
    Similarly, by (\ref{eqn:gsp_tiebreak_H}) and the concavity of $f$, we have
    \begin{equation} \label{gsp_concave_2}
        f(|A_h^\truth|+1) - f(|A_h^\truth|) \geq f(|A_h^\lie|) - f(|A_h^\lie| - 1),
    \end{equation}
    where by Claim~2, if $f(0) = -\infty$ then $|A_h^\lie| - 1 > 0$.
    It follows that
    \begin{align*}
        0 &\geq w_\ell \cdot f(|A_\ell^\truth|-1) - w_\ell \cdot f(|A_\ell^\truth|) + w_h \cdot f(|A_h^\truth|+1) - w_h \cdot f(|A_h^\truth|) \\
        &= w_\ell \cdot [f(|A_\ell^\truth|-1) - f(|A_\ell^\truth|)] + w_h \cdot [f(|A_h^\truth|+1) - f(|A_h^\truth|)]\\
        &\geq w_\ell \cdot [f(|A_\ell^\lie|) - f(|A_\ell^\lie|+1)] + w_h \cdot [f(|A_h^\lie|) - f(|A_h^\lie| - 1)]\\
        &= w_\ell \cdot f(|A_\ell^\lie|) - w_\ell \cdot f(|A_\ell^\lie|+1) + w_h \cdot f(|A_h^\lie|) - w_h \cdot f(|A_h^\lie|-1) \\
        &\geq w_\ell \cdot f(|A_\ell^\lie|) - w_\ell \cdot f(|A_\ell^\lie|+1) + w_b \cdot f(|A_b^\lie|) - w_b \cdot f(|A_b^\lie|-1) \\
        &\geq 0,
    \end{align*}
    where the first inequality holds by (\ref{gsp_concave_truth}), the second inequality by (\ref{gsp_concave_1}) and (\ref{gsp_concave_2}), 
    the third inequality by (\ref{eqn:gsp_hb}),
    and the last inequality by (\ref{gsp_concave_lie}).
    Note also that the first inequality can be an equality only if $\ell < h$, the third inequality only if $h \le b$, and the last inequality only if $b < \ell$. 
    Hence, the three inequalities can become equalities simultaneously only if $\ell < h \le b < \ell$, which is impossible.
    This gives us the final contradiction and completes the proof.
\end{proof}

\section{Conclusion}

In this work, we have shown that for agents with matroid-rank valuations and arbitrary entitlements, the class of weighted additive welfarist rules with 
 concave functions exhibits desirable monotonicity and strategyproofness properties, thereby extending previous results on the maximum weighted Nash welfare (MWNW) rule and binary additive valuations \citep{SuksompongTe22}.
Combined with the results of \citet{MontanariScSu22}, our findings strengthen the case for the maximum weighted harmonic welfare (MWHW) rule and its variants based on modified harmonic numbers, especially in comparison to MWNW, in both the matroid-rank and the binary additive setting (see further discussion in Section~\ref{sec:intro}).
We also show in Appendix~\ref{app:XOS-subadditive} that several of our positive results cease to hold if one moves from matroid-rank valuations to the more general classes of \emph{binary XOS} and \emph{binary subadditive} valuations.

Possible future directions include exploring other classes of valuation functions such as \emph{$2$-value functions} \citep{AmanatidisBiFi21,AkramiChHo22} and \emph{restricted additive valuations} (also called \emph{generalized binary valuations}) \citep{AkramiReSe22,CamachoFePe23}, both of which also extend binary additive valuations,\footnote{In Appendix~\ref{app:restricted-additive}, we show that under restricted additive valuations, even in the unweighted setting, MWNW---which reduces to the maximum Nash welfare (MNW) rule---fails resource- and population-monotonicity as well as (individual) strategyproofness.} and examining weighted additive welfarist rules with non-concave functions.
Investigating monotonicity and strategyproofness properties of allocation rules beyond weighted additive welfarist rules, including more general weighted welfarist rules which are not necessarily additive, or introducing constraints on the permissible allocations \citep{Suksompong21} may reveal interesting insights as well.

\subsection*{Acknowledgments}

This work was partially supported by the Singapore Ministry of Education under grant number MOE-T2EP20221-0001 and by an NUS Start-up Grant.
We thank Vignesh Viswanathan and Yair Zick for helpful discussions, and the anonymous reviewers for thoughtful comments and suggestions.

\bibliographystyle{plainnat}
\bibliography{main}

\appendix

\section{Binary XOS and Binary Subadditive Valuations}
\label{app:XOS-subadditive}

In this appendix, we show that several positive results that we demonstrated for matroid-rank (i.e., binary submodular) valuations cannot be extended to the more general classes of binary XOS and binary subadditive valuations.
These classes have been studied, for example, by \citet{BabaioffEzFe21-dichotomous} and \citet{BarmanVe21-XOS}; we first recall their definitions.

\begin{definition}
A valuation function $v$ is said to be 
\begin{itemize}
\item \emph{binary} if $v(G'\cup\{g\}) - v(G')\in \{0,1\}$ for all $G'\subseteq G$ and $g\in G\setminus G'$;
\item \emph{XOS} (also known as \emph{fractionally subadditive}) if there exists a collection of additive valuation functions $h_1,\dots,h_k$ such that $v(G') = \max_{j=1}^k h_j(G')$ for all $G'\subseteq G$;
\item \emph{subadditive} if $v(G') + v(G'') \ge v(G'\cup G'')$ for all $G', G'' \subseteq G$.
\end{itemize}
The function $v$ is said to be \emph{binary XOS} if it is both binary and XOS, and \emph{binary subadditive} if it is both binary and subadditive.
\end{definition}

It is well-known that every matroid-rank function is binary XOS and every binary XOS function is binary subadditive.
We establish in the subsequent propositions that, even in the unweighted setting, MNW fails resource- and population-monotonicity as well as (individual) strategyproofness, regardless of tie-breaking.

\begin{proposition}
\label{prop:XOS-resmon}
    In the unweighted setting, under binary XOS valuations, MNW is not resource-monotone regardless of tie-breaking.
\end{proposition}

\begin{proof}
    Consider an instance with two agents and ten goods.
    Agent~$1$ has a binary additive valuation~$v_1$ with value~$1$ on each of $g_1,\dots,g_6$ and $0$ on $g_7,\dots,g_{10}$.
    Agent~$2$'s valuation is such that
    \[ 
    v_2(G') = 
    \begin{cases} 
    
      |G'| & \text{ if } |G'| \leq 3; \\
      3 & \text{ if }  |G'| > 3 \text{ and at least one of } g_6,g_{10} \text{ is not in } G'; \\
      4 & \text{ if } |G'| > 3 \text{ and both } g_6,g_{10} \text{ are in } G'. 
   \end{cases}
    \]
    It is clear that $v_2$ is binary; we claim that it is also XOS.
    To this end, for $x,y,z\subseteq\{1,2,\dots,10\}$ such that $x < y < z$, let $h_{x,y,z}$ be a binary additive function with value~$1$ on each of $g_x,g_y,g_z$ and $0$ on all remaining goods.
    Moreover, for $x,y\in\{1,2,3,4,5,7,8,9\}$ such that $x < y$, let $h_{x,y,6,10}$ be a binary additive function with value~$1$ on each of $g_x,g_y,g_6,g_{10}$ and $0$ on all remaining goods.
    One can check that for any $G'\subseteq G$, the value $v_2(G')$ is equal to the maximum among $h_{x,y,z}(G')$ and $h_{x,y,6,10}(G')$ over all functions $h_{x,y,z}$ and $h_{x,y,6,10}$ that we defined.
    This means that $v_2$ is binary XOS.\footnote{On the other hand, note that $v_2$ is not submodular, since $v_2(\{g_1,g_2,g_3,g_6\}) - v_2(\{g_1,g_2,g_3\}) = 3 - 3 < 4 - 3 = v_2(\{g_1,g_2,g_3,g_6,g_{10}\}) - v_2(\{g_1,g_2,g_3,g_{10}\})$.}

    In this instance, one MNW allocation gives $g_1,\dots,g_5$ to agent~$1$ and $g_6,g_7,g_8,g_{10}$ to agent~$2$.
    Moreover, in every MNW allocation, agent~$1$ receives a utility of~$5$ and agent~$2$ a utility of~$4$---indeed, if agent~$1$ receives a utility of~$6$, then $g_6$ must be allocated to agent~$1$, which leaves agent~$2$ with utility no more than~$3$, and we have $6\cdot 3 < 5\cdot 4$.

    Next, consider a modified instance where we remove $g_{10}$.
    In this instance, $v_1$ remains binary additive and $v_2$ binary XOS.
    The unique MNW allocation gives $g_1,\dots,g_6$ to agent~$1$ (for a utility of~$6$) and $g_7,g_8,g_9$ to agent~$2$ (for a utility of~$3$).
    Hence, agent~$1$'s utility increases from $5$ to $6$ upon the removal of $g_{10}$, meaning that MNW fails resource-monotonicity.
\end{proof}

\begin{proposition}
\label{prop:XOS-popmon}
    In the unweighted setting, under binary XOS valuations, MNW is not population-monotone regardless of tie-breaking.
\end{proposition}

\begin{proof}
    Consider the instance with two agents and ten goods described in the proof of Proposition~\ref{prop:XOS-resmon}.
    Recall that agent~$1$ receives a utility of~$5$ in every MNW allocation.

    Suppose we introduce the third agent who has a binary additive valuation~$v_3$ with value~$1$ on~$g_{10}$ and $0$ on all remaining goods.
    In order for the Nash welfare to be positive in the modified instance, $g_{10}$ must be allocated to agent~$3$.
    Given this, the only way to maximize the Nash welfare is to allocate $g_1,\dots,g_6$ to agent~$1$ and $g_7,g_8,g_9$ to agent~$2$.
    Hence, agent~$1$'s utility increases from~$5$ to $6$ upon the introduction of agent~$3$, meaning that MNW fails population-monotonicity.
\end{proof}

\begin{proposition}
\label{prop:XOS-strategyproof}
    In the unweighted setting, under binary XOS valuations, MNW is not strategyproof regardless of tie-breaking.
\end{proposition}

\begin{proof}
    Consider the instance with two agents and ten goods described in the proof of Proposition~\ref{prop:XOS-resmon}.
    Recall that agent~$1$ receives a utility of~$5$ in every MNW allocation.

    Suppose agent~$1$ lies by changing her utility for $g_{10}$ from $0$ to $1$.
    With this false valuation, the unique MNW allocation gives $g_1,\dots,g_6$ as well as $g_{10}$ to agent~$1$ and $g_7,g_8,g_9$ to agent~$2$---this yields a Nash welfare of $7\cdot 3 = 21$ (with respect to the false valuation), which is higher than the Nash welfare of $5\cdot 4 = 20$ that can be obtained by allocating $g_6$ and $g_{10}$ to agent~$2$.     
    In particular, agent~$1$ receives a utility of $6$ with respect to her true valuation.
    Hence, agent~$1$'s utility increases from $5$ to $6$ by misreporting her valuation, meaning that MNW fails strategyproofness.
\end{proof}

Since every binary XOS function is also binary subadditive, we immediately obtain the following corollary.

\begin{corollary}
    In the unweighted setting, under binary subadditive valuations, MNW is not resource-monotone regardless of tie-breaking.
    The same holds for population-monotonicity as well as strategyproofness.
\end{corollary}

We remark that, using the same instances and similar proofs, one can check that all results in this appendix hold for the maximum harmonic welfare (MHW) rule as well.

\section{Restricted Additive Valuations}
\label{app:restricted-additive}

When agents have additive valuations and equal weights, \citet{ChakrabortyScSu21} proved that  MNW fails both resource- and population-monotonicity regardless of tie-breaking.
Moreover, a result by \citet{KlausMi02} implies that MNW also fails (individual) strategyproofness.\footnote{See the discussion in Section~1 in the work of \citet{HalpernPrPs20}.}
In this appendix, we show that these negative results continue to hold even under \emph{restricted} additive valuations.
A valuation profile $(v_1,\dots,v_n)$ is said to be \emph{restricted additive} (or \emph{generalized binary}) if $v_1,\dots,v_n$ are additive and there exists a function $h:G\rightarrow\mathbb{R}_{\ge 0}$ such that $v_i(g)\in \{0,h(g)\}$ for all $i\in N$ and $g\in G$ \citep{AkramiReSe22,CamachoFePe23}.

\begin{proposition}
    In the unweighted setting, under restricted additive valuations, MNW is not resource-monotone regardless of tie-breaking.
\end{proposition}

\begin{proof}
    Consider an instance with three agents and four goods and the following valuations:
    \begin{center}
        \begin{tabular}{ c | c c c c }
         $\mathbf{v}$ & $g_1$ & $g_2$ & $g_3$ & $g_4$ \\ 
         \hline
         $1$ & $0$ & \circled{$1$} & $0$ & \circled{$3$} \\  
         $2$ & \circled{$5$} & $1$ & $0$ & $3$ \\
         $3$ & $5$ & $0$ & \circled{$2$} & $0$ \\
        \end{tabular}
    \end{center}
    The circled cells represent the unique MNW allocation, which has a Nash welfare of $(1+3)\cdot 5\cdot 2 = 40$.
    To see this, note that agent~$3$ is the only one who values $g_3$ positively, so MNW must allocate $g_3$ to agent~$3$.
    If agent~$3$ also receives $g_1$, then $g_2$ and $g_4$ must be split between agents~$1$ and $2$, resulting in a Nash welfare of only $1\cdot 3\cdot (5+2) = 21$.
    Hence, $g_1$ must be allocated to agent~$2$, and agent~$1$ must receive at least one of $g_2$ and $g_4$.
    One can check that if agent~$1$ receives only one of $g_2$ and $g_4$, the resulting Nash welfare is lower than that of the circled allocation.

    Next, suppose we introduce the fifth good:
    \begin{center}
        \begin{tabular}{ c | c c c c | c }
         $\mathbf{v}$ & $g_1$ & $g_2$ & $g_3$ & $g_4$ & $g_5$ \\ 
         \hline
         $1$ & $0$ & $1$ & $0$ & $3$ & \circled{$6$} \\  
         $2$ & $5$ & \circled{$1$} & $0$ & \circled{$3$} & $0$ \\
         $3$ & \circled{$5$} & $0$ & \circled{$2$} & $0$ & $0$ \\
        \end{tabular}
    \end{center}
    The circled cells represent the unique MNW allocation, which has a Nash welfare of $6\cdot (1+3)\cdot (5+2) = 168$.
    To see this, note again that $g_3$ must be allocated to agent~$3$, and similarly $g_5$ must be allocated to agent~$1$.
    Now, $g_1$ can be allocated to either agent~$2$ or $3$, and $g_2$ and $g_4$ can be allocated to either agent~$1$ or $2$.
    One can check that among these allocations, the circled allocation yields the highest Nash welfare.

    Hence, agent~$2$'s utility decreases from $5$ to $4$ upon the introduction of $g_5$, meaning that MNW fails resource-monotonicity.
\end{proof}

\begin{proposition}
    In the unweighted setting, under restricted additive valuations, MNW is not population-monotone regardless of tie-breaking.
\end{proposition}

\begin{proof}
    Consider an instance with two agents and four goods and the following valuations:
    \begin{center}
        \begin{tabular}{ c | c c c c }
         $\mathbf{v}$ & $g_1$ & $g_2$ & $g_3$ & $g_4$ \\ 
         \hline
         $1$ & $5$ & $0$ & \circled{$2$} & \circled{$9$} \\  
         $2$ & \circled{$5$} & \circled{$3$} & $0$ & $9$ \\
        \end{tabular}
    \end{center}
    The circled cells represent the unique MNW allocation, which has a Nash welfare of $(2+9)\cdot (5+3) = 88$.
    To see this, note that agent~$1$ is the only one who values $g_3$ positively, so MNW must allocate $g_3$ to agent~$1$.
    Similarly, $g_2$ must be allocated to agent~$2$.
    One can check that among the possible ways to allocate $g_1$ and $g_4$, the circled allocation yields the highest Nash welfare.

    Next, suppose we introduce the third agent:
    \begin{center}
        \begin{tabular}{ c | c c c c }
         $\mathbf{v}$ & $g_1$ & $g_2$ & $g_3$ & $g_4$ \\ 
         \hline
         $1$ & \circled{$5$} & $0$ & \circled{$2$} & $9$ \\  
         $2$ & $5$ & $3$ & $0$ & \circled{$9$} \\
         \hline
         $3$ & $5$ & \circled{$3$} & $0$ & $0$ \\
        \end{tabular}
    \end{center}
    The circled cells represent the unique MNW allocation, which has a Nash welfare of $(5+2)\cdot 9\cdot 3 = 189$.
    To see this, note that again $g_3$ must be allocated to agent~$1$.
    We consider three cases based on whom $g_1$ is allocated to.
    \begin{itemize}
    \item Suppose $g_1$ is allocated to agent~$1$.
    Then, agent~$3$ must receive $g_2$ in order to avoid getting zero utility.
    As a consequence, agent~$2$ must receive $g_4$ in order to avoid getting zero utility, and we arrive at the circled allocation.
    \item Suppose $g_1$ is allocated to agent~$2$.
    Again, agent~$3$ must receive $g_2$ in order to avoid getting zero utility.
    If $g_4$ is allocated to agent~$1$, the Nash welfare is $(2+9)\cdot 5\cdot 3 = 165$, whereas if $g_4$ is allocated to agent~$2$, the Nash welfare is $2\cdot (5+9)\cdot 3 = 84$.
    \item Suppose $g_1$ is allocated to agent~$3$.
    Agent~$2$ needs at least one of $g_2$ and $g_4$ in order to avoid getting zero utility.
    \begin{itemize}
    \item If agent~$2$ receives both $g_2$ and $g_4$, then the Nash welfare is $2\cdot (3+9)\cdot 5 = 120$.
    \item If agent~$2$ receives only $g_2$, then $g_4$ must be allocated to agent~$1$, and the Nash welfare is $(2+9)\cdot 3\cdot 5 = 165$.
    \item If agent~$2$ receives only $g_4$, then $g_2$ must be allocated to agent~$3$, and the Nash welfare is $2\cdot 9\cdot (5+3) = 144$.
    \end{itemize}
    \end{itemize}
    Hence, agent~$2$'s utility increases from $8$ to $9$ upon the introduction of agent~$3$, meaning that MNW fails population-monotonicity.
\end{proof}

\begin{proposition}
    In the unweighted setting, under restricted additive valuations, MNW is not strategyproof regardless of tie-breaking.
\end{proposition}

\begin{proof}
    Consider an instance with two agents and three goods and the following valuations:
    \begin{center}
        \begin{tabular}{ c | c c c }
         $\mathbf{v}$ & $g_1$ & $g_2$ & $g_3$ \\ 
         \hline
         $1$ & $5$ & \circled{$2$} & \circled{$2$} \\  
         $2$ & \circled{$5$} & $0$ & $2$ \\
        \end{tabular}
    \end{center}
    The circled cells represent the unique MNW allocation, which has a Nash welfare of $(2+2)\cdot 5 = 20$.
    To see this, note that agent $1$ is the only one who values $g_2$ positively, so MNW must allocate $g_2$ to agent~$1$.
    One can check that among the possible ways to allocate $g_1$ and $g_3$, the circled allocation yields the highest Nash welfare.

    Next, suppose agent $1$ lies about her utilities for $g_2$ and $g_3$:
    \begin{center}
        \begin{tabular}{ c | c c c }
         $\mathbf{v}'$ & $g_1$ & $g_2$ & $g_3$ \\ 
         \hline
         $1$ & \circled{$5$} & $0$ & $0$ \\  
         $2$ & $5$ & $0$ & \circled{$2$} \\
        \end{tabular}
    \end{center}
    The circled cells represent one MNW allocation, which has a Nash welfare of $5 \cdot 2 = 10$.
    To see this, note that agent $2$ is the only one who values $g_3$ positively, so MNW must allocate $g_3$ to agent $2$. 
    On the other hand, agent~$1$ must receive $g_1$ in order to avoid getting zero utility. 
    Since no agent values $g_2$ positively, it can be left unallocated (or allocated to either agent). 
    It follows that in every MNW allocation, agent~$1$ receives a utility of~$5$.

    Hence, agent $1$'s utility increases from $4$ to $5$ by misreporting her valuation, meaning that MNW fails strategyproofness.
\end{proof}

\end{document}